\documentclass[runningheads]{llncs}
\usepackage{amssymb,amsmath,multicol}
\usepackage{graphicx,url}
\usepackage{color} %,setspace,enumitem}
\usepackage{epstopdf}

\usepackage[ruled]{algorithm}
\usepackage[noend]{algpseudocode}

\usepackage{xspace}

\floatname{algorithm}{{\small Code}}

\algblockdefx[Receive]{Receive}{EndReceive}%
[1]{{\bf receive}   (#1) from process $\pr$}%
{{\bf end receive}}

\algblockdefx[Upon]{Upon}{EndUpon}%
[1]{{\bf upon} (#1) do}%
{{\bf end upon}}

\makeatletter
\ifthenelse{\equal{\ALG@noend}{t}}%
{\algtext*{EndReceive}}
{}%

\makeatletter
\ifthenelse{\equal{\ALG@noend}{t}}%
{\algtext*{EndUpon}}
{}%

\algrenewcommand{\ALG@beginalgorithmic}{\scriptsize}
\algrenewcommand\alglinenumber[1]{\scriptsize #1:}

\providecommand{\tup}[1]{%
    \relax\ifmmode
      \langle #1 \rangle%
    \else
        $\langle$#1$\rangle$%
    \fi
}

\newcommand{\act}[1]{%
    \relax\ifmmode
        \mathord{\mathcode`\-="702D\sf #1\mathcode`\-="2200}%
    \else
        $\mathord{\mathcode`\-="702D\sf #1\mathcode`\-="2200}$%
    \fi
}

\newcommand{\remove}[1]{}

\makeatletter
\def\mainlistofsymbols{
  \normalsize
  \vspace*{1.5 em}
  \@starttoc{los}
}

\def\partonelistofsymbols{
  \normalsize
  \vspace*{1.5 em}
  \@starttoc{p1los}
}

\def\parttwolistofsymbols{
  \normalsize
  \vspace*{1.5 em}
  \@starttoc{p2los}
}

\def\l@symbol#1#2{\addpenalty{-\@highpenalty} \vskip 4pt plus 2pt
{\@dottedtocline{0}{0em}{8em}{#1}{#2}}}
\makeatother

\newcommand{\newhiddensym}[2]{%
}

\newcommand{\algIOA}[2]{\ifmmode{\text{#1}_{#2}}\else{$\text{#1}_{#2}$}\fi}
\newcommand{\EX}{\ifmmode{\xi}\else{$\xi$}\fi}
\newcommand{\EXF}{\ifmmode{\phi}\else{$\phi$}\fi}
\newcommand{\inter}[1]{
	\ifmmode{\left(\bigcap_{\mathcal{Q}\in#1}\mathcal{Q}\right)}
	\else{$\left(\bigcap_{\mathcal{Q}\in#1}\mathcal{Q}\right)$}
	\fi
}
\newcommand{\ledger}{\mathcal{L}}
\mathchardef\mhyphen="2D
\newcommand{\pr}{p}
\newcommand{\rdr}{r}
\newcommand{\bef}{\rightarrow}

\newcommand{\vid}[1]{\ifmmode{\nu_{#1}}\else{$\nu_{#1}$}\fi}
\newcommand{\seen}{\ifmmode{seen}\else{$seen$}\fi}
\newcommand{\maxts}[1]{\ifmmode{maxTS_{#1}}\else{$maxTS_{#1}$}\fi}
\newcommand{\maxtag}[1]{\ifmmode{maxTag_{#1}}\else{$maxTag_{#1}$}\fi}
\newcommand{\maxpair}[1]{\ifmmode{maxMPair_{#1}}\else{$maxMPair_{#1}$}\fi}
\newcommand{\mintag}[1]{\ifmmode{minTag_{#1}}\else{$minTag_{#1}$}\fi}
\newcommand{\maxps}{\ifmmode{maxPS}\else{$maxPS$}\fi}
\newcommand{\conftg}[1]{\ifmmode{confirmed_{#1}}\else{$confirmed_{#1}$}\fi}
\newcommand{\maxconftag}{\ifmmode{\ms{maxCT}}\else{$maxCT$}\fi}

\newcommand{\extends}{\Vert}

\newcommand{\nn}[1]{\textcolor{green}{#1}}
\newcommand{\dc}[1]{\textcolor{cyan}{#1}}
\newcommand{\cg}[1]{\textcolor{black}{#1}}
\newcommand{\af}[1]{\textcolor{red}{#1}}

\title{Appending Atomically in\\ Byzantine
Distributed Ledgers\thanks{This work was co-funded by the European Regional Development Fund and the Republic of Cyprus through the Research Promotion Foundation (Project: POST-DOC/0916/0090), by research funds from the University of Cyprus (CG-RA2020), by the Spanish grant TIN2017-88749-R (DiscoEdge), the Region of Madrid EdgeData-CM program (P2018/TCS-4499), the NSF of China grant 61520106005, and by the Spanish Ministerio de Educaci\'on Cultura y Deporte under grant PRX18/00163.}
}

\author{Vicent Cholvi\inst{1} \and
Antonio Fern\'andez Anta\inst{2} \and
Chryssis Georgiou\inst{3}
\and\break 
Nicolas Nicolaou\inst{4}
\and
Michel Raynal\inst{5,6}
}
\institute{Universitat Jaume I, Spain\\
\email{vcholvi@uji.es} \and
IMDEA Networks Institute, Spain\\
\email{antonio.fernandez@imdea.org} \and
University of Cyprus, Cyprus\\
\email{chryssis@cs.ucy.ac.cy} \and
Algolysis Ltd, Cyprus\\
\email{nicolas@algolysis.com}
\and
Univ. Rennes IRISA, France\\
\email{michel.raynal@irisa.fr}\and
Dept. of Computing, Polytechnic Univ., Hong Kong
}

\authorrunning{V. Cholvi, A. Fern\'andez Anta, C. Georgiou, N. Nicolaou, and M. Raynal}
\newcommand{\mdledger}{\mathcal{M}}

\newcommand{\append}{{\textsc{append}\xspace}}
\newcommand{\get}{{\textsc{get}\xspace}}
\newcommand{\atomicappends}{{Atomic Appends\xspace}}

\begin{document}

\maketitle 

%\setcounter{footnote}{0}
%TODO mandatory: add short abstract of the document
\begin{abstract}
A Distributed Ledger Object (DLO) 
%as first defined in \cite{DLO_SIGACT18}) 
is a concurrent object that maintains a totally ordered sequence of records,
%(ledger), 
and supports two basic operations: \append, which appends a record at the end of the sequence, and \get, which returns the sequence of records.
In this work we provide a proper formalization of a {\em Byzantine-tolerant Distributed Ledger Object} (BDLO), which is a DLO in a distributed system in which processes may deviate arbitrarily from their indented behavior, i.e. they may be Byzantine. Our formal definition is accompanied by algorithms to implement BDLOs by utilizing an underlying  Byzantine Atomic Broadcast service. 

\noindent{We then utilize the BDLO implementations to solve the {\em \atomicappends{}} problem % (as first 
%defined in \cite{AA2019}) 
against Byzantine processes. The \atomicappends{} problem emerges 
when several clients have records to append, the record of each client has to be appended to a different BDLO,
and it must be guaranteed that either \emph{all} records are appended or \emph{none}. We present distributed algorithms implementing solutions for the \atomicappends{} problem when the clients (which are involved in the appends) and the servers (which maintain the BDLOs) may be Byzantine.}

\vspace{.7em}
\end{abstract}

%\tableofcontents

%\vspace{-5em}

%CG: I turned it into a contributions paragraph and I have structured the paper as suggested. 
\remove{
\noindent\textbf{\af{
AF: Possible structure of the paper:\\
- Present and proof correctness of a Byzantine-tolerant DLO (or BDLO) in which an unbounded number of clients can fail.\\
- Present and proof correctness Byzantine-tolerant DLO in which only a bounded number $t$ of clients can fail ($t$-BDLO), and hence no spurious records are appended.\\
- Present and proof how the two above can be combined to solve the \atomicappends{} problem. The idea is using a SDLO like in the Tokenomics 2019 paper, implemented with a set $S$ of $n$ servers up to which at most $t$ can fail. The DLOs on which the \atomicappends{} is applied are implemented as $t$-BDLO, so only if at least one correct process in $S$ appends in them the append takes place.\\
}}}
%{\bf CG: As we discussed, we will not consider VDLOS, for various reasons.}

\section{Introduction}
\label{sec:Intro}

There has been a great interest recently in the so-called crypto-technologies (e.g., blockchain systems~\cite{N08bitcoin}), and distributed ledger technology (DLT) in general~\cite{Zago2018}, which are becoming very popular and are expected to have a high impact in multiple aspects of our everyday life. Although such a recent popularity is primarily due to the explosive growth of  numerous crypto-currencies, there are many applications of this core technology that are outside the financial industry.  These applications arise from leveraging various useful features provided by distributed ledgers, such as a  decentralized information management,  immutable record keeping for possible audit trail, robustness, availability, security, and privacy (see, for instance,~\cite{DLT:health:2017,Namecoin,DLT:RealEstate:2017,DLT:Science}).
However, there are many different blockchain systems, and new ones are proposed almost everyday. Hence, it is extremely unlikely that one single DLT or blockchain system will prevail. This is forcing the DLT community to accept that it is inevitable to come up with ways to make blockchains interconnect and interoperate. 

In that direction, the work in~\cite{DLO_SIGACT18} proposed a formal definition of a reliable concurrent object, termed Distributed Ledger Object (DLO), which tries to convey the essential elements of blockchains. In particular, a DLO maintains a sequence of records, and has only two operations, {\append} and {\get}. The {\append} operation is used to add a new record at the end of the sequence, while the {\get} operation returns the sequence. Using the above-mentioned formalism, in~\cite{AA2019} the authors initiated the study of systems formed by multiple DLOs that interact among each other. 
Namely, they defined the \emph{\atomicappends{} problem,} in which 
several clients have records to append, the record of each client has to be appended to a different DLO, and it must be guaranteed that either all records are appended or none. \cg{Consider, for example, two clients $A$ and $B$ where the one, say $A$, buys a car from $B$. Record $r_A$ includes the transfer of the car's digital deed from $B$ to $A$,
and $r_B$ includes the transfer from $A$ to $B$  the agreed amount in some digital currency $c$. DLO$_A$ is a ledger maintaining digital deeds and DLO$_B$ maintains transactions in the digital currency $c$. 
So, while the two records are mutually dependent, they concern different DLOs,
hence  the Atomic Append problem requires that {\em either}
record $r_A$ is appended in DLO$_A$ and record $r_B$ is appended in DLO$_B$ {\em or} none of the records are appended in the corresponding DLOs.} 

In~\cite{AA2019} the clients were assumed to be selfish and rational 
%and risk-averse~
\cite{osborne2004introduction}, and could have different incentives for the different outcomes. Additionally, it was assumed that they could fail by crashing, which makes solving the problem more challenging. The authors showed that for some cases the existence of an intermediary is necessary for the problem solution, and proposed the implementation of such intermediary over a specialized blockchain (termed \emph{Smart DLO}, SDLO), also showing how this can be used to solve the \atomicappends{} problem even in an asynchronous, client competitive environment, in which all the clients may crash.\vspace{.5em}

%\section{Related Work}
\noindent{\textbf{Related Work:}}
The \atomicappends{} problem we describe above is very related to the multi-party fair exchange problem \cite{DBLP:conf/fc/FranklinT98}, in which several parties exchange commodities so that everyone gives an item away and receives an item in return. However, the proposed solutions for this problem rely on cryptographic techniques \cite{DBLP:conf/focs/MicaliRK03,DBLP:conf/fc/MukhamedovKR05} and are not designed for distributed ledgers. 

\sloppypar{
Among the first problems identified involving the interconnection of blockchains was the Atomic Cross-chain Swap~\cite{DBLP:conf/podc/Herlihy18}, which can also be seen as a version of the fair exchange problem. In this case, two or more users want to exchange assets (usually cryptocurrency) in multiple blockchains. Herlihy \cite{DBLP:conf/podc/Herlihy18} has formalized and generalized atomic cross-chain swaps beyond one-to-one paths, and shows how multiple cross-chain swaps can be achieved if the transfers form a strongly connected directed graph. Herlihy proves that the best strategy, in Game Theoretic sense, for the users is to
follow the proposed algorithm, and that someone that follows it will never end up worst than at the start.
Unfortunately, these guarantees do not hold if the system is asynchronous.
}

Unlike in most blockchain systems, in Hyperledger Fabric \cite{DBLP:conf/eurosys/AndroulakiBBCCC18,DBLP:conf/esorics/AndroulakiCCK18} it is possible to have transactions that span several blockchains (blockchains are called \emph{channels} in Hyperledger Fabric). This allows solving the atomic cross-chain swap problem using a third trusted channel or a mechanism similar to a two-phase commit \cite{DBLP:conf/esorics/AndroulakiCCK18}. Additionally, these solutions do not require synchrony from the system. The ability of channels to access each other's state and interact is a very interesting feature of Hyperledger Fabric, very in line with the techniques we assume from advanced distributed ledgers in this paper. Unfortunately, they seem to be limited to the channels of a given Hyperledger Fabric deployment.

There are other blockchain systems under development that, like Hyperledger Fabric, will allow interactions between the different chains, presumably with many more operations than atomic swaps. Examples are Cosmos \cite{Cosmos}
%\footnote{\url{https://cosmos.network}, accessed Nov 22, 2018.} 
or PolkaDot \cite{PolkaDot}.
%\footnote{\url{https://polkadot.network}, accessed Nov 22, 2018.}. 
These systems will have their own multi-chain technology, so only chains in a given deployment can initially interact, and other blockchain will be connected via gateways.

The practical need of blockchain systems to access the outside world to retrieve data (e.g., exchange rates, bank account balances) has been solved with the use of \emph{blockchain oracles}. These are relatively reliable sources of data that can be used inside a blockchain, typically in a smart contract. The weakest aspect of blockchain oracles is trust, since the outcome or actions of a smart contract will be as reliable as the data provided by the oracle. As of now, it seems there is no good solution for this trust problem, and blockchains have to rely on oracle services like Oraclize \cite{Oraclize}.\vspace{.5em}
%\footnote{\url{http://www.oraclize.it}, accessed Nov 22, 2018.}.

%\subsection{Contributions}
\noindent{\textbf{Contributions:}}
Contrary to what was assumed in~\cite{DLO_SIGACT18,AA2019} (i.e., both clients and servers can only fail by crashing), in existing blockchain systems, both the servers (e.g., miners) and the clients (e.g., users) could be acting maliciously. To this respect, in this work we present implementations where {\em both} the clients and the servers can be Byzantine, i.e., we present implementations of {\em Byzantine-tolerant} linearizable DLOs. Our contributions are as follows:
\begin{itemize}
    \item We provide a formalization of \emph{Byzantine-tolerant Distributed Ledger Objects} -- BDLOs (Sect.~\ref{sec:model}).
    \item We present and prove the correctness of algorithms that implement a linearizable BDLO (Sect.~\ref{sec:BDLOimpl}) in an asynchronous setting (enriched with a Byzantine Atomic Braoadcast service) in which up to $f$ servers can fail, and\break $(i)$ an unbounded number of clients can fail (Sect.~\ref{subsec:BDLOnobound}); or $(ii)$ only a bounded number $t$ of clients can fail (Sect.~\ref{subsec:BDLObounded}). In the second case we can prevent spurious records to be appended by malicious clients without the need of any additional mechanism.
    \item We provide a definition of the \emph{\atomicappends{}} problem in a system with Byzantine failures (Sect.~\ref{sec:model}).
    \item We present and prove how the above algorithms implementing BDLOs can be combined and adapted to solve the \atomicappends{} problem (Sect.~\ref{sec:AtomicAppends}). (i) First, we build a Smart BDLO (SBDLO, first presented in \cite{AA2019} for tolerating crashes) to aggregate and coordinate the append of multiple records (Sect.~\ref{subsec:SBDLO}). The SBDLO is implemented with a set $N$ of $n \geq 2t+1$ servers up to which at most $t$ can fail. The BDLOs on which the \atomicappends{} is applied are implemented as  BDLOs with a bounded number $t$ of Byzantine clients, so it is guaranteed that only if at least one correct process in $N$ appends in them, the append takes place.\\
    (ii) Then, we show how the problem can be solved by replacing the SBDLO with a ``regular" BDLO and the use of a set $N$ of at least $2t+1$ ``helper'' processes, of which at most $t$ can fail (Sect.~\ref{s:helper}). These processes monitor (by periodic \get\ operations) the BDLO for new \atomicappends{} operations. Once matching \atomicappends{} records are observed, the helper processes perform the \append{} operations to the corresponding BDLOs.
    %\item \cg{Both solutions for the \atomicappends{} problem %require consensus, in order to implement the SBDLO and %BDLO, respectively. We introduce a new concurrent object, %the Byzantine-tolerant Distributed Set object, BDSO, which %operates similarly to DBLO, but maintains a set instead of %a sequence; this object does not require consensus to be %implemented. Hence, using a DBSO and a set of helper %processes we show that the \atomicappends{} problem can be %solved without consensus.} 
\end{itemize}

\section{Model and Definitions}
\label{sec:model}
%\subsection{Distributed Ledger Objects}
\noindent{\textbf{Distributed Ledger Objects:}}
A Distributed Ledger Object (DLO) is a concurrent object that stores a totally ordered sequence of {\em records} (initially empty). A DLO $\ledger$ supports two operations, $\ledger.\append(r)$ and $\ledger.\get()$, which append a new record $r$ to the sequence and return the whole sequence, respectively~\cite{DLO_SIGACT18}. 
{A \emph{record} is a triple $\rdr=\tup{\tau, \pr, v}$, where $\pr$ is the identifier of the process that created record $\rdr$, $v$ is the data of the record drawn from an alphabet $\Sigma$, and
$\tau$ is a {\em unique} record identifier from a set ${\mathcal T}$ {(e.g., the cryptographic hash of $\tup{\pr, v}$).}} 
%$\valSet$.  
%We will use $\rdr.\pr$ to denote the id of the process that created record $\rdr$; similarly we define $\rdr.\tau$ and $\rdr.v$.
%A process $\pr$ invokes an $\ledger.{\get}_\pr()$ operation
%\footnote{We define only one operation to access the value of the ledger for simplicity. In practice, other operations, like those to access individual records in the sequence, will also be available.} to obtain the sequence $\ledger.S$ of records stored in the ledger object $\ledger$, and $\pr$ invokes an $\ledger.{\append}_\pr(\rdr)$ operation to extend $\ledger.S$ with a new record $r$.
%Initially, the sequence $\ledger.S$ is empty. 
%
\remove{
In this paper we will consider DLOs that satisfy 
the following properties,
%\footnote{Other prefix and consistency properties can be defined},
\setlist{nolistsep}
\begin{itemize}[noitemsep, leftmargin=4mm]
\item \emph{Strong prefix} property: Given any two $\ledger.{\get}()$ invocations, which return record sequences $S$ and $S'$ respectively, it must hold that $S$ is a prefix of $S'$ or vice-versa.
\remove{
\item \emph{Linearizability} property: Informally, {\append} and {\get} operations appear as if they occur instantaneously, respecting real-time ordering, and they are consistent with the resulting total order: no ${\get}()$ preceding 
%an operation 
${\append}(r)$ returns a sequence with %record 
$r$, and all {\get} operations that succeed ${\append}(r)$ do.
}
\item \emph{Linearizability} property: Informally, {\append} and {\get} operations appear as if they occur instantaneously, which yields a total order among them. This order must respect real-time ordering, and be consistent with the semantics of the operations: no ${\get}()$ preceding %
%an operation 
${\append}(r)$ returns a sequence with %record 
$r$, and all {\get} operations that succeed ${\append}(r)$ do.
\end{itemize}
}
The DLO is implemented by a set of {\em servers} that collaborate running a distributed algorithm. The DLO is used by a set of {\em clients} that access it by invoking {\append} and {\get} operations, which are translated into request and response messages exchanged with the servers. An execution is a sequence of \emph{invocation} and \emph{return}
events, starting with an invocation event. 
%We assume that 
An operation $\pi$ is \emph{complete} in an execution $\EX$, 
%of the algorithm, 
if both the invocation and matching return 
%events 
of $\pi$ appear in $\EX$. We say that an operation $\pi_1$ \emph{precedes} an operation $\pi_2$, or $\pi_2$ \emph{succeeds} $\pi_1$, in an execution $\EX$ if the return event of $\pi_1$ appears before the invocation event of $\pi_2$ in $\EX$; otherwise the two operations are \emph{concurrent}.
{In this work we focus on {\em linearizable} DLOs~\cite{DLO_SIGACT18}. Informally, under linearizability, the {\append} and {\get} operations appear as if they occur instantaneously, which yields a total order among them. This order must respect real-time ordering, and be consistent with the semantics of operations: no $\get()$ preceding $\append(r)$ returns a sequence with $r$, and all {\get} operations that succeed $\append(r)$ do.}
{By default any client can append records or access the state of a DLO with {\get}. However, if convenient, we may assume that the set of clients that can issue  ({\append} and {\get}) operations is restricted. For instance, we %may 
assume that only the creator $\rdr.\pr$ of a record $\rdr$ can append the record in a DLO $\ledger$, or restrict {\append} operations to a predefined set of clients $N$.}\vspace{.5em}

% Executions of the DLOs are \emph{well formed}, in the sense that each client 
% %(even faulty) 
% invokes a single operation at a time. That is, when a client $c$ invokes an operation $\pi$, $c$ does not invoke an 
% operation $\pi'\neq\pi$ before $\pi$ completes.
%the response of $\pi$ appears in the execution of the algorithm.}
%a client execution is well formed, in the sense that it can only have one operation pending %to be completed (i.e., it has to wait for operations to complete before issuing new ones).

%\subsection{Failure Models}
\noindent{\textbf{Failure Model:}}
In this work we assume that processes (servers and clients) can fail arbitrarily, i.e., we assume that failures are Byzantine. 
%In order to approach the problem of implementing DLOs with Byzantine failures, we define three system models with respect to failures.
%We define the following system model with respect to failures:
%    \item Byzantine-servers system: Clients are reliable, but there can be up to $f$ servers that can fail arbitrarily. The total number of servers is at least $3f+1$.
%    \item Byzantine-clients system: Servers are reliable, but any number of clients can fail arbitrarily.
{Specifically,} we assume a \emph{Byzantine system} in which 
%\emph{any number of clients,} and 
\emph{up to $f$ servers} can fail arbitrarily and that the total number of servers is at least $3f+1$. {For clients we consider two cases: $(i)$ any number of clients can be Byzantine; $(ii)$ up to $t$ clients can be Byzantine.}
{We assume that each process $p$ (client or server) has a pair of public and private keys, and a cryptographic certificate containing its public key. These certificates are generated by a reliable authority, so we discard the possibility of spurious or fake processes (there cannot be Sybil attacks), and have been distributed to all the processes that may interact with each other.
Hence, we} also assume that the messages sent by any process (server or client) are authenticated, so that messages corrupted or fabricated by Byzantine processes are detected and discarded by correct processes~\cite{CRISTIAN1995158}.
Communication channels between correct processes are reliable but asynchronous.
\vspace{.5em}
%clients trying to append the records to the respective DLOs.}

%\subsection{Byzantine-tolerant DLO}
\noindent{\textbf{Byzantine-tolerant DLO:}}
%The objective of 
The {first aim of this paper} is to propose algorithms that implement a {linearizable} DLO $\ledger$ in Byzantine systems.
%, satisfying the strong prefix and linearizability properties. 
%In this paper the objective is to obtain DLOs that satisfy the strong prefix and linearizability properties. 
Since Byzantine clients and servers can behave arbitrarily, 
%we have to adapt these properties 
{we define the properties that a DLO must satisfy adapted}
to Byzantine systems. 
%Let $t(\pi)$ the instant when an operation $\pi$ appears to occur.}
%defining which operations are considered for the definition of these two properties.
{In particular, 
%while the initial definitions consider that processes may return or append \textit{valid} sequences, this is invalidated from 
{since Byzantine processes 
%that 
may return any arbitrary sequence or append any record, the properties only consider the actions of \emph{correct processes}.%\vspace{-1em} 
}
%Therefore, in order to capture the Byzantine behavior, we re-focus the two definitions only considering the actions of \emph{correct processes}. Note that the new definition generalizes and captures the initial definitions.
}
\remove{
\nn{[NN: i feel that the following properties are loosely defined. What do we mean by consider operations? Do we look at the execution and we take out the operation from faulty processes? If we do remove those operations then how a Byzantine process will manage to append a record as we say in bullet 3? Maybe we should be more specific? An attempt follows.]} \af{[AF: Agree]}
\begin{itemize}
    \item All {\get} and {\append} operations issued by correct processes are considered. 
    \item The {\get} operations issued by Byzantine clients are not considered in the evaluation of the properties, since these clients can arbitrarily corrupt the sequence returned.
    \item An {\append} operation of record $r$ issued by a Byzantine client is considered in the evaluation of the properties {\em if and only if} some correct client obtains in the response to a {\get} operation a sequence that contains $r$.
\end{itemize}
}
\begin{itemize}%[leftmargin=4mm]
    \item {\emph{Byzantine Completeness (BC):} All the {\get} and {\append} operations invoked by correct clients eventually complete.}
    \item \emph{Byzantine Strong Prefix (BSP):} If two \textit{correct clients} issue two $\ledger.\get()$
    operations that return record sequences $S$ and $S'$ respectively, then either $S$ is a 
    prefix of $S'$ or vice-versa.
    
% \item \emph{Byzantine Linearizability (BL):} If a correct process $p$ invokes a $\ledger.{\get}()_p$ operation and gets a sequence $S$, then: (i) for all the records $r\in S$, ${\append}(r)$ \af{precedes} 
%     or is concurrent to 
%     the $\ledger.{\get}()_p$ operation, 
%     (ii) for all the $\ledger.{\append}(r)$ operations where $r\in S$ and invoked by correct processes are totally ordered respecting real-time ordering, and (iii) if a correct process $q$ invokes a $\ledger.{\get}()_q$ operation that succeeds $\ledger.{\get}()_p$ then it returns a sequence $S'$ such that $S$ is a prefix of $S'$.
    
\item \emph{Byzantine Linearizability (BL):} Let $G$ be the set of all complete {\get} operations issued by correct clients. Let $A$ be the set of complete {\append} operations $\ledger.{\append}(r)$ such that $r \in S$ and
$S$ is the sequence returned by some operation
$\ledger.{\get}() \in G$. Then linerizability holds with respect to the set of operations
$G \cup A$. This property is similar to the one described in \cite{DBLP:journals/mst/MostefaouiPRJ17} for registers.
\end{itemize}

In the remainder we say that a DLO is {\em Byzantine Tolerant} if it satisfies the properties {BC,} BSP, and BL in a Byzantine system. 
%\cg{along with the above liveness propery;} 
{We will be referring to {a Byzantine-tolerant DLO} by BDLO}.

Observe that it is possible that some \append\ 
operations issued by Byzantine processes might not be distinguished by correct processes. To this respect, we consider the notion of {\em effective appends}. 
To add a record $r$ to a BDLO, a Byzantine process $p_k$ can invoke $\append(r)$, in which case it behaves as if it was correct. It can also attempt to add a record to the BDLO without explicitly invoking an $\append()$. 
      To this  end, it can send underlying messages which, from the point of view of the correct processes, 
      simulate  an invocation of $\append(r)$. It can also intertwine several such insertion of records to the BDLO.
       Such attempts to add records, without invoking the operation $\append()$ may or not succeed. We say that such an attempt constitutes an {\em effective append} if no correct process can distinguish it from a correct
      invocation of $\append()$. Hence an effective append adds a record to the BDLO. We study this issue in Section~\ref{sec:BDLOimpl}.\vspace{.5em}

\remove{
Hence, we define three types of Byzantine-tolerant DLOs:
\begin{itemize}
    \item A DLO is Byzantine-servers-tolerant if it satisfies the strong prefix and linearizability properties in a Byzantine-servers system.
    \item A DLO is Byzantine-clients-tolerant if it satisfies the strong prefix and linearizability properties in a Byzantine-clients system.
    \item A DLO is Byzantine-tolerant if it satisfies the strong prefix and linearizability properties in a Byzantine system.
\end{itemize}
}

\sloppypar{
\noindent
{\bf Multiple DLOs (MDLO) and Multiple BDLOs (MBDLO):}
A {\em Multi-Distributed Ledger Object} $\mdledger$, termed MDLO, consists of a collection $D$ of (heterogeneous {linearizable}) DLOs
and supports the following operations: (i) $\mdledger.{\get}_\pr(\ledger)$, and (ii) $\mdledger.{\append}_\pr(\ledger,\rdr)$~\cite{AA2019}.
The ${\get}$ operation returns the sequence of records $\ledger.S$, where $\ledger\in D$. Similarly, the ${\append}$ operation 
appends the record $\rdr$ to the end of the sequence $\ledger.S$, where $\ledger\in D$.
%Similarly, a MVDLO is a collection of VDLOs. Observe that if 
From the locality property of linearizability~\cite{HW90} it follows
that a MDLO is linearizable, if it is composed of linearizable DLOs.
{\em Multiple BDLOs}, termed MBDLO, are defined similarly over a collection of BDLOs (i.e., {Byzantine-tolerant DLOs})\footnote{Note that we do not restrict whether the BDLOs in the MBDLO are implemented by common servers, or each BDLO is implemented by different servers, as long as the total number of servers and the bound on how many can fail is respected.}.
}
%that satisfy the BSP and BL properties).}
\vspace{.5em}

\noindent{{\bf The \atomicappends{} Problem:} Following~\cite{AA2019}, we define the \atomicappends{} problem, termed \emph{AtomicAppends}, that captures the properties we need to satisfy when multiple operations attempt to append {\em dependent} records on different BDLOs of an MBDLO. 
Intuitively, \emph{AtomicAppends} is analogous to the atomic cross-chain swap \cite{DBLP:conf/podc/Herlihy18} in a MBDLO.
In the crash failure model considered in~\cite{AA2019}, informally} \emph{AtomicAppends} requires that either \emph{all} records will be appended on the involved BDLOs or \emph{none}.
However, in the Byzantine failures model it is impossible from preventing a faulty client to append its record without coordination with the rest of clients. Hence, the \atomicappends{} problem has to be redefined for this failure model.
%\footnote{\cg{In~\cite{AA2019}, the \atomicappends{} problem was defined over MDLOs that tolerate crash failures. The definition can readily be adopted over MBDLOs that tolerate Byzantine failures.}} 

We say that a record $r$ \emph{depends} on a record $r'$, if $r$ may be appended on its intended BDLO, say $\ledger$,
only if $r'$ is appended on a BDLO, say $\ledger'$. 
Two records, $r$ and $r'$, are \emph{mutually dependent}, if $r$ depends on $r'$ and $r'$ depends on $r$.

\begin{definition}[$2$-AtomicAppends]
	\label{def:2AA}
	Consider two clients, $p$ and $q$, with mutually dependent records $r_p$ and $r_q$. 
		We say that records $r_p$ and $r_q$ are {\em appended atomically} in BDLO $\ledger_p$ and BDLO $\ledger_q$, respectively, 
	when:
	\begin{itemize}
		\item {\emph{AA-safety (AAS):} The record (say $r_p$ wlog) of a correct client ($p$) is appended (in $\ledger_p$) only if the record
		of the other client ($q$, which may be correct or not) is also appended (in $\ledger_q$).}
		\item {\emph{AA-liveness (AAL):}} If both $p$ and $q$ are correct, 
		then both records are appended eventually. 
	\end{itemize}
\end{definition}

{As mentioned above, it is not possible to prevent a faulty client $q$ from appending its record $r_q$ even if the correct client $p$ does not. What the safety property AAS guarantees is that the opposite cannot happen. This is analogous of the property in atomic cross-chain swaps \cite{DBLP:conf/podc/Herlihy18} that a correct process cannot end up worse than at the beginning. For instance, when the records represent the transfer of assets between $p$ and $q$, if a faulty client appends its record it is transferring its asset possibly without receiving anything in exchange.}

An algorithm {\em solves} the $2$-AtomicAppends problem under a given {system, if it guarantees the safety and
liveness properties AAS and AAL of Definition~\ref{def:2AA} in every execution of the system.} Since we consider Byzantine failures, our system model with respect to the \atomicappends{} problem is such that the correct processes want to proceed with the append of the records {(to guarantee liveness AAL),} while the Byzantine processes may
%desire to block the transaction (prevent liveness AAL) or 
%\marginpar{\af{Removed ref to liveness because does not apply when any client is Byzantine.}}
try to get correct clients to append (to prevent safety AAS).

%Observe that the safety condition must hold even if one or both $A$ and $B$ crash or do not follow protocol $P$ (i.e., even if they are malicious).
%{\bf CG: The definition of the problem is general wrt the failure model, i.e., covers any type of failure, but for our initial investigation of the problem, as we discussed, it seems natural  to consider rational clients that might fail by crashing.}\\

The $k$-\emph{AtomicAppends} problem, for $k\ge 2$, is a generalization of the $2$-AtomicAppends that can be defined in the natural way ($k$ clients, 
with $k$ mutually dependent records, to be appended to up to $k$ BDLOs.) 
From this point onwards, we will focus on the $2$-AtomicAppends problem, and when clear from the context, 
we will refer to it simply as \emph{AtomicAppends}.\vspace{.5em}

%\subsection{Byzantine Atomic Broadcast}
\noindent{\textbf{Byzantine Atomic Broadcast:}}
In the algorithms we propose in this paper for implementing BDLOs, we use %assume %available and use 
a Byzantine Atomic Broadcast (BAB) service 
%used among 
for the server communication \cite{coelho2018byzantine,CRISTIAN1995158,DBLP:conf/srds/MilosevicHS11}, that satisfies the properties of validity, agreement, integrity, and total order, defined as follows:
\begin{itemize} %[leftmargin=4mm]
	\item 
	\textit{Validity}: if a correct server BAB-broadcasts a message, then it will eventually BAB-deliver it.
	\item
	\textit{Agreement}: if a correct server BAB-delivers a message, then all correct servers will eventually BAB-deliver that message.
	\item
	\textit{Integrity}: a message is BAB-delivered by each correct server at most once, and only if it was previously BAB-broadcast.
	\item
	\textit{Total Order}: the messages BAB-delivered by correct servers are totally ordered; i.e., if any correct server BAB-delivers message $m$ before message $m'$, then every correct server must do it in that order.%\vspace{-1em}
\end{itemize}

Note that the work in \cite{DLO_SIGACT18} utilized a crash-tolerant Atomic Broadcast (AB) service to 
implement a crash-tolerant DLO. The properties assumed here for the BAB service are similar to their counterpart in the AB service, but applied only to correct processes (since in the AB service processes stop when they fail, these properties could be satisfied by the whole set of processes).
It is important to mention that it is not enough to replace the AB service with a
BAB service in the algorithms of \cite{DLO_SIGACT18} to implement a Byzantine DLO, and ensure the satisfaction of properties BC, BSP, and BL. Therefore, we need to introduce some additional machinery.

%{Note that the work in \cite{DLO_SIGACT18} utilized a %crash-tolerant atomic broadcast to 
%implement crash-tolerant DLO. 
%}

\section{Algorithms for Byzantine-tolerant DLOs}
\label{sec:BDLOimpl}
In this section, we introduce algorithms for implementing Byzantine-tolerant DLOs. First, we assume that that there is no bound on the number of clients that can fail (Section~\ref{subsec:BDLOnobound}). However, if all clients can be Byzantine, then there is no way to prevent a client from appending a meaningless record (unless we assume that servers are \emph{clairvoyants}, in the sense of detecting such records simply by checking them). In other words, effective appends are possible (cf. Sect.~\ref{sec:model}). Thus, in Section~\ref{subsec:BDLObounded} we assume that there is a bound $t$ on the maximum number of clients that can fail, and provide the algorithms that implement the corresponding DLOs. In that case, we detect the meaningless records by requesting an append operation to be issued by, at least, $t+1$ clients; hence, effective appends can be prevented.

\subsection{{Unbounded Number of Byzantine Clients}}
\label{subsec:BDLOnobound}

\newcommand{\ubdl}{u-ByDL}

%\noindent\begin{minipage}{\textwidth}
%   \centering
%\nn{[NN: I suggest we call the following algorithm u-ByDL: unbounded Byzantine Distributed Ledger]}

%\noindent
%\begin{minipage}{0.99\textwidth}
    	\begin{algorithm}[t!]
			\caption{\small API to the operations of a BDLO $\ledger$, executed by Client $p$} 
			\label{code:interface-distributed-client}
		    %\begin{multicols}{2}
		        \begin{algorithmic}[1]
		        	%\State \dk{Let $f$ be an upper bound on the number of Byzantine servers}
    				\State  \textbf{Init:} $c \leftarrow 0$
    				\Function{$\ledger.{\get}$}{~}
    				\State $c \leftarrow c + 1$
    				\State {\bf send} request ($c$, $p$, {\sc get}) to %\geq$
    				{at least $2f +1$ different servers}
    				\State \textbf{wait} resp. ($c$, $i$, {\sc getResp}, $S$) from $f +1$ different servers with the same sequence $S$
    				\State \textbf{return} $S$
    				\EndFunction
    %						\columnbreak
    				\Function{$\ledger.{\append}$}{$r$}
    				\State $c \leftarrow c + 1$
    				\State \textbf{send} request ($c$, $p$, {\sc append}, $r$) to at least $2f +1$ different servers
    				\State \textbf{wait} resp. ($c$, $i$, {\sc appendResp}, {\sc ack}) from $f +1$ different servers
    				\State \textbf{return} {\sc ack}
    				\EndFunction
    			\end{algorithmic}
    	    %\end{multicols}
		\end{algorithm}
%\end{minipage}

%\hfill

%\noindent
%\begin{minipage}{0.99\textwidth}
		\begin{algorithm}[t!]
			\caption{\small Algorithm \ubdl{}: Byzantine-tolerant DLO; Code for Server $i$}
			\label{impl:at-server}
			%\begin{multicols}{2}
			\begin{algorithmic}[1]
				\State \textbf{Init:} $S_i \leftarrow \emptyset$%; $pending_i \leftarrow \emptyset$\vspace{-1em}
%							\Statex
				\Receive{$c$, $p$, {\sc get}} 
				\State \act{BAB-broadcast}($c$, $p$, {\sc get}, $i$) 
%				\State add $(p,c)$ to $get\_pending_i$ 
				\EndReceive
				\Upon{\act{BAB-deliver}($c$, $p$, {\sc get}, $j$)}
				\If{(($c$, $p$, {\sc get}, -) has been BAB-delivered %exactly 
				$f+1$ times from different servers)} 
				%\cg{Do we need to say exactly? Isn't it at least?} %\af{I think it has to be exactly. Otherwise there %will be many responses sent}
%				\State \indent \textbf{if} $(p,c) \in get\_pending_i$ \textbf{then}
				\State \textbf{send} resp. ($c$, $i$, {\sc getResp}, $S_i$) to $p$
%				\State \indent \indent remove $(p,c)$ from $get\_pending_i$
				\EndIf
				\EndUpon
				\Receive{$c$, $p$, {\sc append}, $r$} 
				\State \act{BAB-broadcast}($c$, $p$, {\sc append}, $r$, $i$)
%				\State add $(c,r)$ to $pending_i$ 
				\EndReceive
				\Upon{\act{BAB-deliver}($c$, $p$, {\sc append}, $r$, $j$)}
				\If{($r \notin S_i$) and \\
         \hskip\algorithmicindent\hspace{0.2cm}  (($c$, $p$, {\sc append}, $r$, -) has been BAB-delivered %exactly 
				from $f+1$ different servers)} 
				\State $S_i \leftarrow S_i  \extends r$
				
%				\If {$\exists (c,r) \in pending_i$}
				\State  \textbf{send} resp. ($c$, $i$, {\sc appendResp}, {\sc ack}) to $p$
%				\State  remove $(c,r)$ from $pending_i$ 
%				\vspace{-1em}
%				\EndIf
				\EndIf
				\EndUpon
			\end{algorithmic}
			%\end{multicols}
		\end{algorithm}
%    \end{minipage}
%\end{minipage}
%\vspace{1em}

%\subsection{Client Algorithm}
\noindent{\textbf{Client Algorithm:}}
The algorithm executed by a client that invokes a {\get} or {\append} operation on a DLO $\ledger$ is presented in Code~\ref{code:interface-distributed-client}. 
%In this code, the maximum number of faulty servers $f$ is assumed to be known and used. (Observe that in the Byzantine-clients system model $f=0$.) 
% Of course, a Byzantine client does not need to follow this algorithm, and can in principle behave arbitrarily.
%
%For the purpose of the linearizability property satisfaction, 
An operation starts with the \textbf{invocation} (event) of the corresponding function in Code~\ref{code:interface-distributed-client}, and it ends when the matching \textbf{return} instruction is executed (return event). 
A Byzantine client $p$ may not follow Code~\ref{code:interface-distributed-client} (as it may behave arbitrarily) but still be able to append a record $r$ in the ledger (with an effective append). So, some correct client may obtain, in the response to a {\get} operation, a sequence that contains a record $r$ appended by a Byzantine client.
%In this case, this append operation has to be considered for linearizability (as mentioned above). 
%[NN: we mentioned start and end of an operation above the following was a repetition]
% We assume that the operation starts when Client $p$ sends the (first) message (-, $p$, {\sc append}, $r$) to any server, and ends when it gets {\sc ack} responses from $f+1$ different servers.%\vspace{-.1em}

When an operation is invoked, 
%With the algorithm of Code~\ref{code:interface-distributed-client} 
a correct client increments a local counter and then 
%simply 
sends operation requests to a set of at least $2f+1$ servers, to guarantee that at least $f+1$ correct servers receive it. A {\get} operation completes when the client receives $f+1$ \textit{consistent} replies and an {\append} completes when the client receives $f+1$ replies from different servers. Both cases guarantee the response from at least one correct server.

% These correct servers will guarantee that the operation is completed. Then, the client will know that the operation has in fact been completed when it receives a consistent response from $f+1$ different servers (of which at least one is correct).

%\subsection{Server Algorithm}
\noindent{\textbf{Server Algorithm:}}
The algorithm executed by the servers is presented in Code~\ref{impl:at-server}. 
{We denote it Algorithm \ubdl{} (from unbounded Byzantine Distributed Ledger).}
%Observe again that $f$ is known and used, and that for the Byzantine-clients system model, $f=0$. 
The algorithm uses the Byzantine Atomic Broadcast service to impose a total order in the messages shared among the servers. Operations received from clients are BAB-broadcast using this service, which are eventually BAB-delivered. An operation is processed by a server only when it has been BAB-delivered
$f+1$ times (sent by different servers).
This implies that at least one correct server sent it. The properties of the BAB service guarantee that all correct servers receive the same sequence of messages BAB-delivered, and hence process the operations at the same point, maintaining their states consistent.

%By the above discussion and the properties guaranteed by the BAB Service we conclude:\vspace{-.5em}

\begin{theorem}
\label{ubydl-correct}
Algorithm \ubdl{} implements a linearizable Byzantine Tolerant Distributed Ledger Object.
%The combination of the algorithms presented in Codes  \ref{code:interface-distributed-client} and \ref{impl:at-server} implements a linearizable Byzantine Tolerant Distributed Ledger Object.%\vspace{-1em}
\end{theorem}

\begin{proof}
To proof the correctness of the algorithm we need to show that it satisfies both the liveness {property BC} and the safety properties BSP and BL. 

\paragraph{Liveness:} The algorithm guarantees {the liveness property BC} with respect to the failure model we assume. More precisely, each correct {client} sends requests to $2f+1$ servers for an operation $\pi$ and waits for $f+1$ servers to reply. Given that the channels are reliable and up to $f$ servers may fail (and thus not reply) then at least $f+1$ correct servers will eventually receive and BAB-broadcast the request from $\pi$. According to the BAB-Valitidy property, each message broadcasted by a correct server will eventually be BAB-delivered. Furthermore, by the BAB-Agreement property, all correct servers will eventually BAB-deliver the messages broadcasted by correct servers. Thus, at least $f+1$ correct servers will deliver at least $f+1$ messages broadcasted by the correct {client}. Those servers will reply to {operation} $\pi$, and hence the {client} will receive at least $f+1$ replies and terminate.  

\paragraph{Safety:} To prove safety we need to show that any execution of our algorithm satisfies properties BSP and BL.

\noindent {\bf BSP:} Byzantine Strong Prefix (BSP) requires that if two {\get} operations from two correct clients return sequences $S$ and $S'$ resp., then either $S$ is a prefix of $S'$ or $S'$ is a prefix of $S$. To derive contradiction let as assume that $S$ is not a prefix of $S'$. Let $S=r_1r_2\ldots r_n$ and $S'=r_1'r_2'\ldots r_m'$ with $m\geq n$. As $S'$ is not a prefix of $S$, then $\exists r_i$ in $S$, for $1\leq i \leq n$ s.t. $r_i\neq r_i'$. From the algorithm it follows that the {\get} operations received $S$ and $S'$ from at least one correct server as each get operations waits for $f+1$ \emph{different} servers to reply with the \emph{same} sequence. Let $s$ be the correct server that sent $S$ and $s'$ be the correct server that replied with $S'$. Before appending a record $r_j$ in its local ledger, a correct server needs to wait for $f+1$ messages that contain $r_j$ to be BAB-delivered. This guarantees that at least a single correct server received the request for appending $r_j$ and BAB-broadcasted that record. According however to BAB-Agreement if $s$ BAB-delivers $r_j$ then $s'$ will BAB-deliver $r_j$ as well. Furthermore, for each record $r_k$, for $1\leq k\leq j$, that is BAB-delivered in $s$ will also BAB-delivered in $s'$ and according to the BAB-Total Order those records will be delivered in the same order in both correct servers. This can be seen with a simple induction. The first record of $S$, $r_1$, will be BAB-delivered to both servers $s$ and $s'$ (by BAB-Agreement property). Record $r_2$ will be BAB-delivered after $r_1$ in $s$. By BAB-Agreement property $r_2$ will be BAB-delivered to $s'$ as well and by BAB-Total Order $r_2$ cannot be delivered before $r_1$. So by the delivery of $r_2$ to $s$ and $s'$, both servers contain the sequence $r_1r_2$. Suppose this is true up to record $r_k$, for $k<n$, i.e. both servers contain sequence $r_1\ldots r_k$ after the BAB-delivery of $r_k$. As noted before, record $r_{k+1}$ will be BAB-delivered to both servers $s$ and $s'$ and by the total order property $r_{k+1}$ cannot be delivered before any record $r_j$, with $j\leq k$. Thus, after the BAB-delivery of $r_{k+1}$ both servers will contain the sequence $r_1\ldots r_k,r_{k+1}$. By the induction it follows that, after the BAB-delivery of $r_n$ to both $s$ and $s'$, they contain sequences $r_1\ldots r_n$. However this is sequence $S$. Furthermore any record $r_m$, for $m>n$, that is BAB-delivered to $s'$ will be placed after $r_n$ in its local sequence. Thus, $S$ is a prefix of $S'$ and that contradicts our initial assumption. With similar reasoning we may show that if $S$ is longer than $S'$ then $S'$ be a prefix of $S$.

\noindent {\bf BL:} Finally, Byzantine Linearizability (BL) requires that: (i) {\append} operations are ordered with respect to all other operations, (ii) if a {\get} operation returns a sequence that contains a record $r_j$ then an $\append(r_j)$ operation preceded that {\get}, and (iii) if a {\get} operation completes before the invocation of another {\get} operation, i.e. ${\get}_1\bef {\get}_2$, then ${\get}_1$ returns a sequence $S_1$ that is a prefix of the sequence returned by ${\get}_2$, say $S_2$. The total ordering of the {\append} operations is ensured by the BAB service. In particular, if ${\append}(r_1)$ happens before ${\append}(r_2)$, i.e. ${\append}(r_1)\bef {\append}(r_2)$, and both are executed by correct processes, then they will send the append message to $2f+1$ servers, out of which at least $f+1$ correct servers will receive and BAB-broadcast the append. Each correct server will BAB-deliver those appends by BAB-Validity nad BAB-Agreement properties. Thus, each correct server will BAB-deliver at least $f+1$ messages for both records and will reply to the operations. Since ${\append}(r_1)\bef {\append}(r_2)$ then there exists a correct server that replies to ${\append}(r_1)$ before terminating. That server added $r_1$ in its sequence before receiving, and thus appending $r_2$. Therefore, by BAB-Total Order all the correct servers will append $r_1$ before $r_2$ in their local sequences, proving this way (i). As for point (ii) a {\get} operation obtains a sequence that contains a record $r$ only if that sequence is received from at least a single correct server $s$. Thus, since $s$ appended $r$ in its sequence then an ${\append}(r)$ operation must have executed before or concurrently with the {\get} operation. Finally, for two operations ${\get}_1$ and ${\get}_2$, s.t. ${\get}_1\bef {\get}_2$, it holds that the correct server, say $s$, that replies to ${\get}_1$ has delivered and replied to all {\append} operations with records in $S_1$ before BAB-delivering $f+1$ messages BAB-broadcasted for ${\get}_1$. Since, ${\get}_1\bef {\get}_2$, the message for ${\get}_2$ will be BAB-delivered to $s$ after the the delivery of the message of ${\get}_1$ at $s$. Since $s$ is a corect server, then by BAB-Agreement and BAB-Total Order, all the correct servers will BAB-deliver the messages from ${\get}_1$ before the messages from ${\get}_2$. It holds also, that all the correct servers delivered all the records appended before ${\get}_1$, before the delivery of the messages from ${\get}_2$ as well. Thus, ${\get}_2$ will receive a sequence $S_2$ longer or the same size as $S_1$. From the proof of BSP though it follows that $S_1$ is a prefix of $S_2$ and that completes the proof.
\qed
%
% \begin{itemize}
%     \item Consistency on the sequence is guaranteed by the Total Order property on the messages of BAB (We probably need to use a short claim by induction)
%     \item By waiting $f+1$ messages of a particular requests to be delivered, each server guarantees that at least one \emph{correct} server received the get request from the process.
%     \item Finally the process waits for $f+1$ replies with the same sequence to guarantee that the sequence is coming from a correct server. Each process will receive $f+1$ replies with the same sequence as the process sends messages to $2f+1$ so at least $f+1$ correct servers will receive it and BAB broadcast it. By the Validity and Agreement properties all the correct servers will deliver the messages broadcasted by the $f+1$ correct servers, and by the Total Order property they will deliver those messages in the same order. So at least $f+1$ correct servers will receive the broadcasted messages and thus will reply to the {\get} operation. Note also that they will send the same sequence $S$.
%     \item For the {\append} operation we guarantee that at least one correct process received the message, broadcasted the append, delivered $f+1$ broadcasts and finally  replied to the operation. 
% \end{itemize}
%
\end{proof}

\subsection{{Bounded Number of Byzantine Clients}}
\label{subsec:BDLObounded}

\newcommand{\bbdl}{b-ByDL}

%\nn{[NN: I suggest we call the following algorithm b-ByDL: bounded Byzantine Distributed Ledger]}

Observe that DLOs are oblivious to the syntax and semantics of the records they hold~\cite{DLO_SIGACT18}. Hence, in general (and in particular in Section~\ref{subsec:BDLOnobound}), we do not care about the records appended by Byzantine clients. Hence, the above algorithm does not prevent a Byzantine client from performing an effective append that adds a meaningless record $r$ on the DBLO, which may be syntactically or semantically invalid.

In this section we assume that at most $t$ clients can be Byzantine\footnote{{Recall that Sybil attacks are not possible.}}, and prevent these spurious records. This is achieved by having valid records to be appended by 
several clients. It is hence assumed that a valid record $r$ is appended by a set $N$ of at least $2t+1$ clients that invoke the operation $\append(r)$ using Code~\ref{code:interface-distributed-client} in parallel. In this section we do not go into how these clients agree on appending the same record\footnote{The next section shows a scenario where this is guaranteed.}. 
%(We assume that all clients append the same record $r$ for simplicity; if convenient, a more sophisticated scheme could be designed, in which the clients append shares of the record $r$, so that $r$ is recovered and appended from $t+1$ shares \cite{DBLP:journals/cacm/Shamir79,Blakley1979}.)
The processing of the append messages at the servers has to be changed as described in Code~\ref{impl:at-server-bounded}, {which presents the Algorithm \bbdl{} (from bounded Byzantine Distributed Ledger).}

%\noindent
%\begin{minipage}{0.99\textwidth}
		\begin{algorithm}[t!]
			\caption{\small Algorithm \bbdl{}: Byzantine-tolerant BDLO with bounded number of Byzantine clients; Code for processing the {\append} operation at Server $i$}
			\label{impl:at-server-bounded}
			%\begin{multicols}{2}
			\begin{algorithmic}[1]
				\State \textbf{Init:} $S_i \leftarrow \emptyset$%; $pending_i \leftarrow \emptyset$\vspace{-1em}
%							\Statex
				
				\Receive{$c$, $p$, {\sc append}, $r$} 
				\State \act{BAB-broadcast}($c$, $p$, {\sc append}, $r$, $i$)
%				\State add $(c,r)$ to $pending_i$ 
				\EndReceive
				\Upon{\act{BAB-deliver}($c$, $p$, {\sc append}, $r$, $j$)}
				\If{($r \in S_i$)}
				\State  \textbf{send} response ($c$, $i$, {\sc appendResp}, {\sc ack}) to $p$
				\Else				
				\If{(($c$, -, {\sc append}, $r$, -) has been BAB-delivered
          from $f+1$ different servers\\
         \hskip\algorithmicindent\hspace{0.8cm} and received from a set $C$ of $t+1$ different clients)} 
				\State $S_i \leftarrow S_i  \extends r$
				\State  \textbf{send} response ($c$, $i$, {\sc appendResp}, {\sc ack}) to all $q \in C$
				\EndIf
				\EndIf
				\EndUpon
			\end{algorithmic}
			%\end{multicols}
		\end{algorithm}
%    \end{minipage}
%\end{minipage}
%\vspace{1em}

\begin{theorem}
\label{bbydl-correct}
Algorithm \bbdl{} implements a linearizable Byzantine Tolerant Distributed Ledger Object that only contains records appended by correct clients.%\vspace{-1em}
\end{theorem}

\begin{proof}
    To prove \bbdl{} correctness, we need to show that satisfies both liveness and safety properties of a BDLO with the special requirement that any record is appended by a correct client. 

    \paragraph{Livenesss:} {We prove that property BC is satisfied.} With similar arguments as in the proof of Theorem \ref{ubydl-correct} we can show that {an {\append} request issued by a correct client will be received by at least $f+1$ correct servers, and hence}
    each correct server will BAB-deliver an append message from {at least} $f+1$ servers. In addition, since we assume that $2t+1$ clients issue append requests for the same record $r$, then each correct server will receive at least $t+1$ requests for $r$. Thus, correct servers will reply to each client requesting the append and hence each correct client will receive at least $f+1$ replies and terminate. 
    
    \paragraph{Safety:} Following the proof of Theorem \ref{ubydl-correct} we can show that \bbdl{} satisfies both BSP and BL properties. What remains to show is that any record appended in the ledger is sent by a correct client. This follows from the fact that at least $t+1$ correct clients issue append requests for the same record $r$. Given that the communication channels are reliable, those messages will eventually be received by all correct servers. Since the servers wait to receive  $t+1$ append requests for server $r$ then they ensure that at least a single correct client requested $r$ to be appended. Hence, any record on the DL was appended by a correct client and this completes the proof. 
    \qed
\end{proof}

%Issue: We are assuming that there is no possibility of Sybil attack. [Or that it can be extended up to $t$ copies.]

\section{Byzantine \atomicappends{}}
\label{sec:AtomicAppends}

\newcommand{\baadl}{BAADL}

In this section we face the \atomicappends{} problem {in a system where clients and servers may be Byzantine}. For simplicity, we first consider the {$2$-AtomicAppends problem, where two clients, $p$ and $q$, attempt to append atomically two mutually dependent records $r_p$ and $r_q$, in BDLOs $\ledger_p$ and $\ledger_q$, respectively.} 
{In the rest of this section we assume that BDLOs $\ledger_p$ and $\ledger_q$ use Algorithm~\bbdl{} to tolerate up to $t$ Byzantine clients and $f$ Byzantine servers, and only accept \append{} operations from a known set $N$ of at least $2t+1$ clients, of which at most $t$ can fail (hence, effective appends are prevented)}.
% version with two clients, $2$-AtomicAppends. Hence, we first consider that there are two clients, $p$ and $q$, with mutually dependent records $r_p$ and $r_q$, to be appended atomically in BDLOs $\ledger_p$ and $\ledger_q$, respectively.

%\noindent
%\begin{minipage}{0.99\textwidth}
    	\begin{algorithm}[t!]
			\caption{\small API for for the $2$-\act{AtomicAppend} of records $r_p$ and $r_q$ in ledgers $\ledger_p$ and $\ledger_q$ by clients $p$ and $q$, respectively, using SBDLO $\ledger$. Code for Client $p$.} 
			\label{code:sdlo-client-1}
		    %\begin{multicols}{2}
		        \begin{algorithmic}[1]
    				\Function{$\act{AtomicAppends}$}{$p, \{p,q\}, r_p, \ledger_p, r_q)$}
    				\State $\ledger.{\append}(\tup{\tau, p, v})$, where $v=\tup{p, \{p,q\}, r_p, \ledger_p, r_q}$
    				\State \textbf{return} {\sc ack}
%				\Upon{receipt {\rm AppendAck} from SBDLO $\ledger$}
%				\State \textbf{return} 
%				\EndUpon
    				\EndFunction
    			\end{algorithmic}
    	    %\end{multicols}
		\end{algorithm}
%\end{minipage}

%\nn{[NN: I suggest we call the following algorithm BAADL (reminds battle or bottle :) ): Byzantine \atomicappends{} Distributed Ledger]}

\subsection{\atomicappends{} Using a Smart BDLO}
\label{subsec:SBDLO}

As proposed in \cite{AA2019}, in order to coordinate the individual appends we will use a Smart BDLO $\ledger$, that is a special BDLO to which clients $p$ and $q$ delegate the task of appending their records in the respective ledgers. They do that by appending in the SBDLO a description of the \atomicappends{} operation to be completed, as shown in Code \ref{code:sdlo-client-1}.
Client $p$ uses the \append{} operation to provide the SBDLO with the data it requires to complete the \atomicappends{}, namely the participants in the \atomicappends{}, the record $r_p$, the BDLO $\ledger_p$, and the
record $r_q$ the other client is appending.
%, i.e., $h_q=\mathit{HASH}(r_q)$, where $\mathit{HASH}()$ is an appropriate cryptographic summary. It is assumed that the summary $h_q$ has been provided by $q$ to $p$, and that $p$ cannot reconstruct $r_q$ from it\footnote{In a practical setting it may be needed to have additional guarantees from $h_q$. For instance, a mechanism (timestamp or nonce) could be introduced to prevent it to be reused, or it may contain a (zero-knowledge) proof that $r_q$ is correct.}. We avoid assuming that $q$ sends the record $r_q$ to $p$ because that would require additional mechanisms to prevent $p$ from directly appending it to $\ledger_q$, should it be Byzantine.

The SBDLO $\ledger$ is a BDLO with unbounded number of faulty clients but that \emph{only allows the creator of a record to append it.} $\ledger$ is implemented with a set $N$ of at least $2t+1$ servers, out of which at most $t$ may be Byzantine. Hence, the {\append} operation in the client side (Line 2 in Code \ref{code:sdlo-client-1}) is implemented as described in Code \ref{code:interface-distributed-client}, with $t$ instead of $f$ as the maximum number of faulty servers.

%\noindent
%\begin{minipage}{0.99\textwidth}
		\begin{algorithm}[t!]
			\caption{\small Algorithm \baadl{}: Byzantine-tolerant Smart SBDLO; Only the code for the \append{} operation is shown; Code for Server $i$}
			\label{code:sdlo-server}
			%\begin{multicols}{2}
			\begin{algorithmic}[1]
				\State \textbf{Init:} $S_i \leftarrow \emptyset$%; $pending_i \leftarrow \emptyset$\vspace{-1em}
%							\Statex
%				\Receive{$c$, $p$, {\sc get}} 
%				\State \textbf{send} response ($c$, $i$, {\sc getResp}, $\bot$) to $p$
%				\EndReceive
				\Receive{$c$, $p$, {\sc append}, $r$} 
				\State \act{BAB-broadcast}($c$, $p$, {\sc append}, $r$, $i$)
%				\State add $(c,r)$ to $pending_i$ 
				\EndReceive
				\Upon{\act{BAB-deliver}($c$, $p$, {\sc append}, $r$, $j$)}
				\If{($r \notin S_i$) and \\
         \hskip\algorithmicindent\hspace{0.2cm} (($c$, $p$, {\sc append}, $r$, -) has been BAB-delivered %exactly 
				from $t+1$ different servers)} 
				\State $S_i \leftarrow S_i  \extends r$
				\If {$r.v=\tup{p, \{p,q\}, r_p, \ledger_p, r_q}$ and
				$\exists r' \in S_i : r'.v=\tup{q, \{p,q\}, r_q, \ledger_q, r_p}$}
				%such that $h_p=\mathit{HASH}(r_p)$ and $h_q=\mathit{HASH}(r_q)$}
    				\State $\ledger_p.{\append}(r_p)$ \label{appendrp}
    				\State $\ledger_q.{\append}(r_q)$			\label{appendrq}	
%				\State  remove $(c,r)$ from $pending_i$ 
%				\vspace{-1em}
				\EndIf
				\State  \textbf{send} response ($c$, $i$, {\sc appendResp}, {\sc ack}) to $p$
				\EndIf
				\EndUpon
			\end{algorithmic}
			%\end{multicols}
		\end{algorithm}
%\end{minipage}

Code \ref{code:sdlo-server} describes the \append{} operation of Algorithm \baadl{} (from Byzantine \atomicappends{} Distributed Ledger) that implements the SBDLO (the rest of the algorithm is as in Code \ref{impl:at-server}).
As expected, it is very similar to the implementation of a BDLO without restrictions in the number of Byzantine clients, but with a difference. 
%First, clients are not allowed to access the state of the ledger, which is implemented by having the {\get} operations return a especial value $\bot$. The reason is that we do not want to allow a Byzantine client (may even be one not involved in the \atomicappends{} operation) to obtain a record from the SBDLO and append it directly, possibly violating the safety property AAS. Second, 
Every time a record $r$ is added to the sequence $S_i$, it is checked whether a matching record $r'$ is already there. This is the case if $r.v=\tup{p, \{p,q\}, r_p, \ledger_p, r_q}$, and $r'.v=\tup{q, \{p,q\}, r_q, \ledger_q, r_p}$.
%, and $h_p=\mathit{HASH}(r_p)$ and $h_q=\mathit{HASH}(r_q)$. 
If so, the corresponding append operations are issued in the respective BDLOs $\ledger_p$ and $\ledger_q$. 

As mentioned above,
each of the ledgers $\ledger_p$ and $\ledger_q$ are BDLOs with a known, {bounded} set $N$, of at least $2t+1$ clients ({which are the servers implementing the SBDLO $\ledger$}), out of which at most $t$ can be Byzantine. 
These ledgers are implemented in a system of at least $2f+1$ servers out of which at most $f$ can be Byzantine, as presented in Algorithm~\bbdl{} (Code \ref{impl:at-server-bounded}). Hence, a record is appended only if at least $t+1$ clients from $N$ issue append operations of the record. {Notice that unlike in the case of ad-hoc clients, in the case of SBDLO at least $t+1$ correct SBDLO servers will receive the requests by the external clients $p$ and $q$ and will issue the same \append{} operation in ledgers $\ledger_p$ and $\ledger_q$, making bounded BDLOs a practical system.} Moreover, Line 2 of Code \ref{impl:at-server-bounded} is modified to verify that a client $p$ attempting to append is in fact in the set $N$ of authorized clients.

\begin{theorem}
\label{baadl-correct}
The combination of the API of Code \ref{code:sdlo-client-1} and the Algorithm \baadl{} solves the $2$-AtomicAppends problem.
\end{theorem}

\begin{proof}
Let us first prove the liveness property AAL. Consider two correct clients $p$ and $q$ with records $r_p$ and $r_q$, to be appended atomically in BDLOs $\ledger_p$ and $\ledger_q$, respectively. Since it is correct, eventually $p$ will issue the call $\act{AtomicAppends}(p, \{p,q\}, r_p, \ledger_p, r_q)$, which from Code \ref{code:sdlo-client-1} will trigger $\ledger.{\append}(\tup{\tau, p, v})$, with $v=\tup{p, \{p,q\}, r_p, \ledger_p, r_q}$.
From Code~\ref{code:interface-distributed-client} (with $t$ instead of $f$) and the process of the append messages in Algorithm \baadl{}, eventually all the correct servers $i$ of the SBDLO will insert $\tup{\tau, p, v}$ in their sequences $S_i$. Similarly,
eventually all the correct servers $i$ of the SBDLO will insert $\tup{\tau', q, v'}$ with $v'=\tup{q, \{q,p\}, r_q, \ledger_q, r_p}$ in their sequences $S_i$.

Let us consider one such server $i$, and assume wlog that $\tup{\tau, p, v}$
%$\langle p, \{p,q\}, r_p, \ledger_p, h_q \rangle$ 
is inserted fist in $S_i$. Then, as soon as 
$\tup{\tau', q, v'}$
%$\langle q, \{q,p\}, r_q, \ledger_q, h_p \rangle$ 
is also inserted, the condition in Line 9 of Code \ref{code:sdlo-server} holds, and
the operations $\ledger_p.{\append}(r_p)$ and $\ledger_q.{\append}(r_q)$ are issued. Since the SBDLO is implemented with at least $2t+1$ servers out of which at most $t$ are Byzantine, at least $t+1$ servers will issue these {\append} operations. Hence, from Theorem~\ref{bbydl-correct},
$r_p$ and $r_q$ will be appended to BDLOs $\ledger_p$ and $\ledger_q$, respectively.

We now prove the safety property AAS. Let us assume to reach a contradiction that AAS is not satisfied because, wlog, the record $r_p$ of correct client $p$ is appended in $\ledger_p$ while $r_q$ is never appended in $\ledger_q$. Observe that $\ledger_p$ is a BDLO implemented with Algorithm~\bbdl{}, which requires at least $t+1$ different clients appending the same record for the record to be in fact appended.
There are two possibilities depending on who are these processes that append $r_p$ in $\ledger_p$: they are (1) SBDLO servers or (2) they include processes that are not SBDLO servers. Let us consider each case separately.

In Case (1), there are at least $t+1$ SBDLO servers that append $r_p$ in $\ledger_p$. Then, at least one is correct, and does it by executing Lines \ref{appendrp} and \ref{appendrq} in Code~\ref{code:sdlo-server}. But then, all correct servers of SBDLO execute these lines, and since there are at least $t+1$ correct servers, record $r_q$ is also appended in $\ledger_q$, which is a contradiction.

In Case (2), by assumption only the set $N$ of servers of the SBDLO are allowed to issue append operation in $\ledger_p$, and any append message sent by a process not in $N$ will be rejected (recall that messages are authenticated). Hence, this case is not possible.
    \qed
\end{proof}

Code \ref{code:sdlo-client-1} and the Algorithm \baadl{} are easily generalized to $k$-AtomicAppends. In Line 2 of Code \ref{code:sdlo-client-1} the client $p$ sends the set of $k$ clients appending records, and the $k-1$ records appended in addition to $r_p$. Similarly, in Line 9 of Code~\ref{code:sdlo-server} the condition becomes that all $k$ records to be appended are already in $S_i$. If so, all of them are appended in the $k$ corresponding BDLOs.

\subsection{\atomicappends{} Using a BDLO and a Set of Helper Processes}
\label{s:helper}

While using a Smart BDLO solves the \atomicappends{} problem as described above, it requires to implement a DLO that is aware of the contents of the records that are appended into it. This is at conflict with the initial spirit of the
DLO definition, that meant to be a data structure that was oblivious to the records syntax and semantics. In this section we describe how in fact the SBDLO can be replaced by a regular BDLO implemented with Algorithm \ubdl{} (Code \ref{impl:at-server}) and a set $N$ of at least $2t+1$ helper processes, of which at most $t$ can fail.

From the point of view of clients $p$ and $q$, the new approach is transparent. Still they execute Code \ref{code:sdlo-client-1} to issue an \atomicappends{} operation, with the difference that now ledger $\ledger$ is not ``smart" anymore, but a regular Byzantine tolerant DLO (e.g., implemented with Code \ref{impl:at-server}).
Similarly, from the point of view of ledgers $\ledger_p$ and $\ledger_q$ the new approach is transparent, except that now their set $N$ of legal clients to append in them is the set of helper processes described above.

%\noindent
%\begin{minipage}{0.99\textwidth}
		\begin{algorithm}[t!]
			\caption{\small Algorithm used by a helper process to complete \atomicappends{} operations; Code for process $x$}
			\label{impl:no-sdlo}
			%\begin{multicols}{2}
			\begin{algorithmic}[1]
				\State \textbf{Init:} $O_x \leftarrow \emptyset$
				\Loop
				\Comment{Loop forever; execute loop body periodically}
				\State $S_x \leftarrow \ledger.\get{}()$
				\While {$\exists r, r' \in S_x \setminus O_x :r.v=\tup{p, \{p,q\}, r_p, \ledger_p, r_q} \land r'.v=\tup{q, \{p,q\}, r_q, \ledger_q, r_p}$}
    				\State $\ledger_p.{\append}(r_p)$ \label{appendrp-helper}
    				\State $\ledger_q.{\append}(r_q)$			\label{appendrq-helper}	
 				    \State $O_x \leftarrow O_x \cup \{r, r'\}$
				\EndWhile
				\EndLoop
			\end{algorithmic}
			%\end{multicols}
		\end{algorithm}
%\end{minipage}

Hence, the main difference is in the helper processes in set $N$. These processes are continuously running a loop that monitors $\ledger$ for new \atomicappends{} operations to complete. This process is described in Code \ref{impl:no-sdlo}.
As can be seen there, a helper process $x$ periodically issues a \get{} operation on $\ledger$ to obtain its latest contents. Then it checks if it contains pairs of matching \atomicappends{} records that correspond to operations that have not been completed yet. (Observe that $x$ maintains a set $O_x$ of records from $\ledger$ that have been already
used.) If so, it issues the corresponding \append{} operations to complete them.

The proof that this new approach solves the AtomicAppends problem is almost verbatim to the proof of Theorem \ref{baadl-correct}, and it is omitted.

\newcommand{\ubso}{u-ByDS}

%%%%%%%%%%%%%%%%%%%%%%% REMOVING BDSO %%%%%%%%%%%%%%%%%
\remove{
\subsection{\atomicappends{} using a Byzantine-tolerant Set Object, and a Set of Helper Processes}

Finally, we observe that the Smart BDLO and the regular BDLO used in the two previous schemes in fact do not require the records appended to be ordered. This means that their state can be a set of records instead of a totally ordered sequence of records. This observation is not trivial, since it implies that they do not require to solve consensus to maintain it, unlike the previous described Byzantine-tolerant DLOs (that use the Byzantine Atomic Broadcast service, which is equivalent to consensus).

Hence, we define a new concurrent object that we call a Byzantine-tolerant Distributed Set Object (BDSO). To maintain the presentation simple, we assume that a BDSO has the same two operations \append{} and \get{} than a BDLO, but that the latter returns a set instead of a sequence. Hence, the state of a BDSO is the set of records that have been appended to (inserted into) it. A BDSO is defined in similar terms as a BDLO was, but dealing with sets instead of sequences. Hence, it must satisfy the same properties Byzantine Completeness (BC), Byzantine Strong Prefix (BSP, but changing prefix to subset), and Byzantine  Linearizability (BL).

Replacing the BDLO $\ledger$ used in Section \ref{s:helper} to solve AtomicAppends with a BDSO does not change the logic and the behavior of the approach. Hence, the AtomicAppends problem can be solved using a BDSO and a set $N$ of  helper processes.

Code \ref{impl:dso-server} presents Algorithm \ubso{} 
that implements a BDSO using a Byzantine Reliable Broadcast (BRB) service \cite{DBLP:journals/iandc/Bracha87,RaynalBook18}. 
%This service has two primitives \act{BRB-broadcast} and \act{BRB-deliver}, which are analogous to those of Byzantien Atomic Broadcast. The BRB service satisfies the same properties as Byzantine  Atomic  Broadcast except Total Order. 
As mentioned above, the BRB service can be implemented in a Byzantine asynchronous distributed system with up to $f$ faulty processes as long as the total number of processes is at least $3f+1$ \cite{RaynalBook18}. Observe that Consensus cannot be solved in such a system. Hence, the following corollary.

\begin{corollary}
The AtomicAppends problem can be solved in a Byzantine asynchronous distributed system without the need of a Consensus (or Atomic Broadcats) service.
\end{corollary}

Of course, it is quite possible that the ledgers $\ledger_p$ and $\ledger_q$ where records $r_p$ and $r_q$ have to be appended in a 2-AtomicAppends instance need Consensus to be implemented.

%\noindent
%\begin{minipage}{0.99\textwidth}
		\begin{algorithm}[t!]
			\caption{\small Algorithm \ubso{}: Byzantine-tolerant DSO; Code for Server $i$}
			\label{impl:dso-server}
			%\begin{multicols}{2}
			\begin{algorithmic}[1]
				\State \textbf{Init:} $S_i \leftarrow \emptyset$%; $pending_i \leftarrow \emptyset$\vspace{-1em}
%							\Statex
				\Receive{$c$, $p$, {\sc get}} 
				\State \act{BRB-broadcast}($c$, $p$, {\sc get}, $i$) 
%				\State add $(p,c)$ to $get\_pending_i$ 
				\EndReceive
				\Upon{\act{BRB-deliver}($c$, $p$, {\sc get}, $j$)}
				\If{(($c$, $p$, {\sc get}, -) has been BRB-delivered %exactly 
				$f+1$ times from different servers)} 
				%\cg{Do we need to say exactly? Isn't it at least?} %\af{I think it has to be exactly. Otherwise there %will be many responses sent}
%				\State \indent \textbf{if} $(p,c) \in get\_pending_i$ \textbf{then}
				\State \textbf{send} resp. ($c$, $i$, {\sc getResp}, $S_i$) to $p$
%				\State \indent \indent remove $(p,c)$ from $get\_pending_i$
				\EndIf
				\EndUpon
				\Receive{$c$, $p$, {\sc append}, $r$} 
				\State \act{BRB-broadcast}($c$, $p$, {\sc append}, $r$, $i$)
%				\State add $(c,r)$ to $pending_i$ 
				\EndReceive
				\Upon{\act{BRB-deliver}($c$, $p$, {\sc append}, $r$, $j$)}
				\If{($r \notin S_i$) and (($c$, $p$, {\sc append}, $r$, -) has been BRB-delivered %exactly 
				from $f+1$ different servers)} 
				\State $S_i \leftarrow S_i  \cup \{r\}$
				
%				\If {$\exists (c,r) \in pending_i$}
				\State  \textbf{send} resp. ($c$, $i$, {\sc appendResp}, {\sc ack}) to $p$
%				\State  remove $(c,r)$ from $pending_i$ 
%				\vspace{-1em}
%				\EndIf
				\EndIf
				\EndUpon
			\end{algorithmic}
			%\end{multicols}
		\end{algorithm}
%    \end{minipage}
%\end{minipage}
%\vspace{1em}
}
%%%% END OF REMOVE %%%%%

\section{Conclusions}
In this work we formalized the notion of a Byzantine Tolerant Distributed Ledger Object (BDLO) and proposed algorithms implementing such objects in distributed settings where a subset of clients and servers may be Byzantine. We demonstrated the utility of our BDLO implementations by providing  solutions to the \atomicappends{} problem, where clients have mutually dependent records to be appended,
 the record of each client has to be appended to a different BDLO,
and either all records are appended or none. 

Our formalization of BDLOs requires a strong prefix property, which prevents the existence of more than one sequence at any point in time (i.e., no ``forks'' allowed, as termed in the blockchain literature). As shown in~\cite{DBLP:conf/spaa/AnceaumePLPP19,DLO_SIGACT18}, this property requires consensus. Therefore, it would be interesting to investigate more relaxed (weaker) versions of this property (that might not require consensus) and study the guarantees than can be provided within our framework.

%\cg{En route, we have introduced a new concurrent object, the %Byzantine-tolerant Distributed Set Object (BDSO), which %operates similarly to BDLO, but it maintains a set, rather %than a sequence. With the use of a BDSO, we show that the %\atomicappends{} problem can be solved without consensus. [If %space permits, we could list a few directions for future %research.]}

\newpage

\remove{
We believe that it may be worth exploring different configurations when we have DLOs and a SDLO:

- Each ledger uses a disjoint set of servers.

- All ledgers (DLOs and SDLO) use the same set of servers.

- The DLOs use the same set but the SDLO use a different set of servers.

Does it make sense to have failure detectors of Byzantine nodes? We need to think about it. Seems to be interesting for a future work.
[If I recall, there was a Byzantine failure detector proposed some time ago. I will look for it.]
} % end remove

\remove{
\section{Algorithms for Byzantine-tolerant Validated DLOs}
{\bf We discussed that VDLO seems to have some issues to consider, so for now, lets leave it out.}

If instead of a DLO we have a Validated DLO (VDLO), each record can be checked for validity before being appended. The Code 2 changes in the sense that in Line 10 no need to wait for $f+1$ servers delivering the same record. It is enough to receive the record the first time and check for validity before appending.

When solving \atomicappends{} with VDLOs, it is still needed that the VDLO received $t+1$ appends for the same valid record from different clients (the servers implementing the SDLO).

} % end remove

%%%%% START REMOVAL %%%%%
\remove{
---- Antonio version above here

\cg{The above specification is an adaptation of the one  from~\cite{coelho2018byzantine}.}  

\subsection{System Model}

We have $n\geq 3f+1$ servers where up to $f$ might be Byzantine. We consider an asynchronous system (but constraint under the conditions needed to implement the
Byzantine atomic broadcast service). At first we assume that clients are crash-prone only (not Byzantine); this means that they run correctly Code 1 unless they crash.  At the end of the paper, we can argue that we can also tolerate Byzantine clients, that is, clients that are deviating from Code 1, without getting into the semantics of the records. In such a case we will need VDLOs (our Byzantine-tolerant implementation of DLO can trivially used to implement Byzantine-tolerant VDLOs).

\section{Basic implementation}
\label{sec:basic}

We assume that the number of servers is unknown. We also assume that clients are aware of the faulty nature of servers, and it is known (an upper bound on) the maximum number of Byzantine servers $f$. 

\section{Atomic consistency}

The  interface is presented in Code~\ref{code:interface-distributed}. As it can be seen there, every operation request is sent to at least $f+1$ servers, to guarantee that at least one correct server receives and processes the request. \cg{You need to send it to at least $2f+1$ to ensure that at least $f+1$ correct nodes get the request and hence you get $f+1$ replies. If you send it only to $f+1$, and $f$ are the malicious ones and they simply choose not to reply, then you will not get the needed $f+1$ replies. } Moreover, the return value is obtained from $f+1$ servers with the same response: at least one correct server, and the rest  with the same (valid) response. In order to differentiate from different responses, all operations (and their requests and responses) are uniquely numbered with counter $c$: once a client returns a value for each $c$ (i.e., after $f+1$ responses with the same value), the rest of the responses will be identified and ignored. \dc{Furthermore, it must be taken into account that Byzantine servers can confuse correct ones by forwarding appropriately altered messages on behalf of correct servers. Therefore, we assume the messages exchanged during a broadcast are authenticated, so that messages corrupted \cg{or fabricated} by Byzantine servers can be recognized and discarded by correct servers~\cite{CRISTIAN1995158}.} \cg{In addition, for a record to be appended in the DLO, it must be signed by the client that proposed it; this prevents malicious servers to fabricate records.}

Specification of Byzantine Atomic Bcast essentially taken from \cite{coelho2018byzantine} combined with the original definition in the DLO paper:
\begin{itemize} 
	\item 
	\textit{Validity}: if a correct server broadcasts a message, then it will eventually deliver it.
	\item
	\textit{Agreement}: if a correct server delivers a message, then all correct servers will eventually deliver that message.
	\item
	\textit{Integrity}: a message is delivered by each correct server at most once, and only if it was previously broadcast.
	\item
	\textit{Total Order}: the messages delivered by correct processes are totally ordered; that is, if any correct server delivers message $m$ before message $m'$, then every correct server must do it in that order.%\vspace{-1em}
\end{itemize}

\begin{figure}
		\begin{algorithm}[H]
			\caption{\small External Interface of a Distributed Ledger Object $\ledger$ Executed by a Process $p$}
			\label{code:interface-distributed}
				\begin{multicols}{2}
			\begin{algorithmic}[1]
				\State \dk{Let $f$ be an upper bound on the number of Byzantine servers}
				\State  \textbf{Init:} $c \leftarrow 0$
				\Function{$\ledger.{\get}$}{~}
				\State $c \leftarrow c + 1$
				\State {\bf send} request (c, {\sc get}) to\dk{, at least,\cg{ $2f +1$ } servers}
				\State \textbf{wait} response (c, {\sc getRes}, $V$) from \dk{$f +1$ different servers with the same $V$}
				\State \textbf{return} $V$
				\EndFunction
						\columnbreak
				\Function{$\ledger.{\append}$}{r}
				\State $c \leftarrow c + 1$
				\State \textbf{send} request (c, {\sc append}, $r$) to\dk{, at least,\cg{ $2f +1$ } servers}
				\State \textbf{wait} response (c, {\sc appendRes}, $res$) from \dk{$f +1$ different servers with the same $res$}
				\State \textbf{return} $res$
				\EndFunction
			\end{algorithmic}
				\end{multicols}
		\end{algorithm}%\vspace{-2em}
\end{figure}

\dk{We say that atomic broadcast can be \emph{reduced} to consensus provided it can be implemented by using consensus objects. In~\cite{DBLP:conf/srds/MilosevicHS11}, the authors show that, if servers may fail in a Byzantine fashion,  atomic broadcast can be reduced to \emph{Byzantine consensus}. Since Byzantine consensus objects are available~\cite{10.1007/978-3-540-92221-6_4}, then the algorithm presented in Code~~\ref{impl:at} can be also implemented. As a matter of fact, the implementation in~\cite{10.1007/978-3-540-92221-6_4} doesn't need to know the number of servers, but only an upper bound on the maximum number of Byzantine servers.}

\cg{CG: The result of~\cite{DBLP:conf/srds/MilosevicHS11} is perfect for us. In order to make
	things completely clear, I suggest we specify precisely the Atomic Broadcast problem
	under Byzantine servers, and the corresponding consensus problem that it can be reduced to (my 
	understanding, it is the one referred to as strong validity consensus below).}

\dk{Note that, in~\cite{DBLP:conf/srds/MilosevicHS11}, it shown that, while some variants of Byzantine consensus are harder than atomic broadcast and, therefore, can be used to implement it (e.g., the one considered in this paper, which is referred as \emph{strong validity consensus}), other variants  are not sufficient to solve atomic broadcast (e.g., the \emph{weak unanimity consensus}).}

Note also that, in order to  guarantee byzantine consensus, at least $3 f +1$ servers are needed. Therefore, Code~\ref{impl:at} also requires, at least, $3 f +1$ servers. 
  
 \begin{figure}
		\begin{algorithm}[H]
			\caption{\small Atomic Distributed Ledger; Code for Server $i$}
			\label{impl:at}
			\begin{multicols}{2}
			\begin{algorithmic}[1]
				\State \textbf{Init:} $S_i \leftarrow \emptyset$; 
				$get\_pending_i \leftarrow \emptyset$;
				$pending_i \leftarrow \emptyset$\vspace{-1em}
							\Statex
				\Receive{$c$, {\sc get}} 
				\State $\act{ABroadcast}(get,p,c)$ 
				\State add $(p,c)$ to $get\_pending_i$ 
				\EndReceive
				\Upon{$\act{ADeliver}(get,p,c)$}
				\State \indent \textbf{if} $(p,c) \in get\_pending_i$ \textbf{then}
				\State \indent \indent \textbf{send} response ($c$, {\sc getRes}, $S_i$) to $p$
				\State \indent \indent remove $(p,c)$ from $get\_pending_i$
				\EndUpon
				\Receive{$c$, {\sc append}, $r$} 
				\State $\act{ABroadcast}(append, r)$
				\State add $(c,r)$ to $pending_i$ 
				\EndReceive
				\Upon{$\act{ADeliver}(append, r)$}
				\If{$r \notin S_i$} 
				\State $S_i \leftarrow S_i  \extends r$
				
				\If {$\exists (c,r) \in pending_i$}
				\State {\footnotesize \textbf{send} response ($c$, {\sc appendRes}, {\sc ack}) to $r.p$}
				\State  remove $(c,r)$ from $pending_i$ 
				\vspace{-1em}
				\EndIf
				\EndIf
				\EndUpon
			\end{algorithmic}
			\end{multicols}
		\end{algorithm}\vspace{-1em}
\end{figure}

\begin{figure}
		\begin{algorithm}[H]
			\caption{\small Atomic Distributed Ledger with signatures; Code for Server $i$}
			\label{impl:at}
			\begin{multicols}{2}
			\begin{algorithmic}[1]
				\State \textbf{Init:} $S_i \leftarrow \emptyset$; 
				$pending_i \leftarrow \emptyset$\vspace{-1em}
							\Statex
				\Receive{$c$, {\sc append}, $r$, $cert_p$} 
				\State $\act{ABroadcast}(append, r, cert_p)$
				\State add $(c,r)$ to $pending_i$ 
				\EndReceive
				\Upon{$\act{ADeliver}(append, r, cert_p)$}
				\If{($r \notin S_i$) and (record $r$ is correctly signed by client $r.p$ with certificate $cert_p$)} 
				\State $S_i \leftarrow S_i  \extends r$
				
				\If {$\exists (c,r) \in pending_i$}
				\State {\footnotesize \textbf{send} response ($c$, {\sc appendRes}, {\sc ack}) to $r.p$}
				\State  remove $(c,r)$ from $pending_i$ 
				\vspace{-1em}
				\EndIf
				\EndIf
				\EndUpon
			\end{algorithmic}
			\end{multicols}
		\end{algorithm}\vspace{-1em}
\end{figure}

\begin{figure}
		\begin{algorithm}[H]
			\caption{\small Atomic Distributed Ledger with majority; Code for Server $i$}
			\label{impl:at}
			\begin{multicols}{2}
			\begin{algorithmic}[1]
				\State \textbf{Init:} $S_i \leftarrow \emptyset$; 
				$pending_i \leftarrow \emptyset$\vspace{-1em}
							\Statex
				\Receive{$c$, {\sc append}, $r$} 
				\State $\act{ABroadcast}(append, r, i)$
				\State add $(c,r)$ to $pending_i$ 
				\EndReceive
				\Upon{$\act{ADeliver}(append, r, j)$}
				\If{($r \notin S_i$) and (record $r$ has been A-Delivered exactly $f+1$ times from different servers)} 
				\State $S_i \leftarrow S_i  \extends r$
				
				\If {$\exists (c,r) \in pending_i$}
				\State {\footnotesize \textbf{send} response ($c$, {\sc appendRes}, {\sc ack}) to $r.p$}
				\State  remove $(c,r)$ from $pending_i$ 
				\vspace{-1em}
				\EndIf
				\EndIf
				\EndUpon
			\end{algorithmic}
			\end{multicols}
		\end{algorithm}\vspace{-1em}
\end{figure}

We first proof that no record fabricated and broadcasted by a malicious server will be inserted in the sequence $S$ of any correct process. 

\begin{lemma}
    In any execution $\ex$ of Algorithm Let $r'$ be a record s.t.: (i) $r'$ was not appended by any client $c$, and (ii) a server $s'$ invoked \act{ABroadcast}(append, r', *). Then no correct server $s$ will append $r'$ in $S_s$. 
\end{lemma}

\begin{theorem}
	The combination of the algorithms presented in Codes  \ref{code:interface-distributed} and \ref{impl:at} implements an atomic distributed ledger.
\end{theorem}
\begin{proof}
Similar to the SIGACT one.
\end{proof}

\paragraph{Some toughs on the next steps:} 

\begin{enumerate}
\item
The approach taking in this paper, as it is now, is \emph{permissioned}.  The type of ledger's membership is governed by the type of membership allowed by the consensus model.~\cite{DBLP:conf/osdi/CastroL99}. 

Maybe, we should  specify how $\act{ABroadcast}()$ and $\act{ADeliver}()$ behave, in terms of the properties of the consensus algorithm used to implement them  (e.g., PBFT~\cite{DBLP:conf/osdi/CastroL99}, XFT~\cite{199398}, etc).
\cg{CG: I believe this moves along the same lines that I propose above.}

\item
 I think that Hyperledger Fabric fills quite well into our approach. Currently there are (at least) two versions, depending on the used consensus protocol: \emph{Kafka} for crash failures, and \emph{PBFT} for Byzantine failures.
 
\item \cg{CG: Since the implementation in~\cite{10.1007/978-3-540-92221-6_4} does not need
	to know the number of servers, can't we extend the solution (with the appropriate adjustments of course) to permissionless blockchains? Would that be very hard?} 
  
 \end{enumerate}

Place a remark that in a similar manner we can implement Byzantine-tolerant versions of Sequential and Eventual consistent DLOs.
}
%%%%% END REMOVAL %%%%%

\remove{
\section{Sequential Consistency}

\begin{figure}[t]
	%\setlength{\columnsep}{5pt}
	%\begin{figure}[t]
	\begin{multicols}{2}
		\begin{algorithm}[H]
			\caption{\small Sequentially Consistent Distributed Ledger; Code for Server $i$}
			\label{impl:sqc}
			%\begin{multicols}{2}
			\begin{algorithmic}[1]
				\State \textbf{Init:} $S_i \leftarrow \emptyset$; $pending_i \leftarrow \emptyset$; $get\_pending_i \leftarrow \emptyset$ %\vspace{-1em}
				%		\Statex
				%		\State \textbf{Upon receiving} request ({\sc get}) from client $p$ do
				%		\State \textbf{send} response ({\sc getRes}, $\ledger$) to $p$
				%		\Statex
				%		\State \textbf{Upon receiving} request ({\sc append}, $r$) from client $p$ do
				%		\State $\act{ABroadcast}(r)$
				%		\State \textbf{add} $r$ to $pending$
				%		\Statex
				\Receive{c, {\sc get}, $\ell$}
				\If{$|S_i| \geq \ell$} 
				\State \textbf{send} response (c, {\sc getRes}, $S_i$) to $p$%\vspace{-1em}
				\Else
				\State \textbf{add} $(c,p,\ell)$ to $get\_pending_i$%\vspace{-1em}
				\EndIf
				\EndReceive
				%		\Statex
				\Receive{c, {\sc append}, $r$} 
				\State $\act{ABroadcast}(c, r)$
				\State \textbf{add} $(c,r)$ to $pending_i$%\vspace{-1em}
				\EndReceive
				%		\Statex
				\Upon{$\act{ADeliver}(c,r)$}
				%\State $\ledger \leftarrow \ledger  \extends r$
				\If{$r \notin S_i$} $S_i \leftarrow S_i  \extends r$
				\EndIf
				\If{$(c,r) \in pending_i$}
				\State \textbf{send} resp. (c, {\sc appendRes}, {\sc ack}, $|S_i|$) to $r.p$
				\State \textbf{remove} $(c,r)$ from $pending_i$
				\EndIf
				\If{$\exists (c',p,\ell) \in get\_pending_i: |S_i| \geq \ell$}
				\State \textbf{send} response ($c'$, {\sc getRes}, $S_i$) to $p$
				\State \textbf{remove} $(c',p,\ell)$ from $get\_pending_i$
				\EndIf
				\EndUpon
			\end{algorithmic}
			%\end{multicols}
		\end{algorithm}
		
		\begin{algorithm}[H]
			\caption{\small External Interface for Sequential Consistency Executed by a Process $p$}
			\label{code:interface-distributed-sqc}
			%	\begin{multicols}{2}
			\begin{algorithmic}[1]
				\State \dk{Let $f$ be an upper bound on the number of Byzantine servers}
				\State \textbf{Init:} $c \leftarrow 0$; $\ell_{last} \leftarrow 0$
				%\Statex
				\Function{$\ledger.{\get}$}{~}
				\State $c \leftarrow c + 1$
				\State {\bf send} request ($c$, {\sc get}, $\ell_{last}$)  to\dk{, at least, \cg{$2f +1$ } servers}
				\State \textbf{wait} response ($c$, {\sc getRes}, $V$) from \dk{$f +1$ different servers with the same $V$}	
				%		such that 
				%		\State \index \index $(V_{last}$ is a prefix of $V) \land( (r_{last} \in V) \lor (r_{last} = \bot) )$
				\State $\ell_{last} \leftarrow |V|$	
				\State \textbf{return} $V$%\vspace{-1em}
				\EndFunction
				%\Statex
				%		\columnbreak
				\Function{$\ledger.{\append}$}{r}
				\State $c \leftarrow c + 1$
				\State \textbf{send} request ($c$, {\sc append}, $r$) to\dk{, at least,\cg{ $2f +1$}  servers}
				\State \textbf{wait} response ($c$, {\sc appendRes}, $res$, $pos$) from \dk{$f +1$ different servers with the same $res$ and $pos$}	
				\State $\ell_{last} \leftarrow pos$
				\State \textbf{return} $res$
				\EndFunction
			\end{algorithmic}
			%	\end{multicols}
		\end{algorithm}
	\end{multicols}\vspace{-1em}
\end{figure}

\begin{theorem}
	The combination of the algorithms presented in Codes~\ref{code:interface-distributed-sqc} and~\ref{impl:sqc} implements a sequentially distributed ledger.
\end{theorem}
\begin{proof}
Similar to the SIGACT one.
\end{proof}
}
%\bibliographystyle{plain}
%\bibliography{biblio}
%\setlength{\bibsep}{0.0pt}
%\setlength{\bibitemsep}{0.0pt}

\bibliographystyle{plainurl}% the mandatory bibstyle
\bibliography{biblio}

\begin{thebibliography}{10}

\bibitem{DBLP:conf/spaa/AnceaumePLPP19}
Emmanuelle Anceaume, Antonella~Del Pozzo, Romaric Ludinard, Maria
  Potop{-}Butucaru, and Sara {Tucci Piergiovanni}.
\newblock Blockchain abstract data type.
\newblock In {\em The 31st {ACM} on Symposium on Parallelism in Algorithms and
  Architectures, {SPAA} 2019, Phoenix, AZ, USA, June 22-24, 2019}, pages
  349--358. {ACM}, 2019.

\bibitem{DBLP:conf/eurosys/AndroulakiBBCCC18}
Elli Androulaki, Artem Barger, Vita Bortnikov, Christian Cachin, Konstantinos
  Christidis, Angelo~De Caro, David Enyeart, Christopher Ferris, Gennady
  Laventman, Yacov Manevich, Srinivasan Muralidharan, Chet Murthy, Binh Nguyen,
  Manish Sethi, Gari Singh, Keith Smith, Alessandro Sorniotti, Chrysoula
  Stathakopoulou, Marko Vukolic, Sharon~Weed Cocco, and Jason Yellick.
\newblock Hyperledger fabric: a distributed operating system for permissioned
  blockchains.
\newblock In Rui Oliveira, Pascal Felber, and Y.~Charlie Hu, editors, {\em
  Proceedings of the Thirteenth EuroSys Conference, EuroSys 2018, Porto,
  Portugal, April 23-26, 2018}, pages 30:1--30:15. {ACM}, 2018.
\newblock URL: \url{http://dl.acm.org/citation.cfm?id=3190508}.

\bibitem{DBLP:conf/esorics/AndroulakiCCK18}
Elli Androulaki, Christian Cachin, Angelo~De Caro, and Eleftherios
  Kokoris{-}Kogias.
\newblock Channels: Horizontal scaling and confidentiality on permissioned
  blockchains.
\newblock In Javier L{\'{o}}pez, Jianying Zhou, and Miguel Soriano, editors,
  {\em Computer Security - 23rd European Symposium on Research in Computer
  Security, {ESORICS} 2018, Barcelona, Spain, September 3-7, 2018, Proceedings,
  Part {I}}, volume 11098 of {\em Lecture Notes in Computer Science}, pages
  111--131. Springer, 2018.

\bibitem{DLO_SIGACT18}
Antonio~Fern{\'{a}}ndez Anta, Kishori~M. Konwar, Chryssis Georgiou, and
  Nicolas~C. Nicolaou.
\newblock Formalizing and implementing distributed ledger objects.
\newblock {\em {SIGACT} News}, 49(2):58--76, 2018.

\bibitem{DLT:Science}
S~Bartling and B~Fecher.
\newblock Could blockchain provide the technical fix to solve sciences
  reproducability, crisis?
\newblock London School of Economics Impact of Social Sciences blog.
  http://blogs.lse.ac.uk/impactofsocialsciences/2016/07/21/could-blockchain-provide-the-technical-fix-to-solve-sciences-reproducibility-crisis/
  (last accessed February 10, 2018.

\bibitem{coelho2018byzantine}
Paulo Coelho, Tarcisio~Ceolin Junior, Alysson Bessani, Fernando Dotti, and
  Fernando Pedone.
\newblock Byzantine fault-tolerant atomic multicast.
\newblock In {\em DSN 2018}, pages 39--50. IEEE, 2018.

\bibitem{CRISTIAN1995158}
F.~Cristian, H.~Aghili, R.~Strong, and D.~Dolev.
\newblock Atomic broadcast: From simple message diffusion to byzantine
  agreement.
\newblock {\em Information and Computation}, 118(1):158 -- 179, 1995.

\bibitem{AA2019}
Antonio {Fernández Anta}, Chryssis Georgiou, and Nicolas Nicolaou.
\newblock Atomic appends: Selling cars and coordinating armies with multiple
  distributed ledgers.
\newblock In {\em International Conference on Blockchain Economics, Security
  and Protocols (Tokenomics 2019)}, pages 39--50, Paris, France, 2019.

\bibitem{DBLP:conf/fc/FranklinT98}
Matthew~K. Franklin and Gene Tsudik.
\newblock Secure group barter: Multi-party fair exchange with semi-trusted
  neutral parties.
\newblock In Rafael Hirschfeld, editor, {\em Financial Cryptography, Second
  International Conference, FC'98, Anguilla, British West Indies, February
  23-25, 1998, Proceedings}, volume 1465 of {\em Lecture Notes in Computer
  Science}, pages 90--102. Springer, 1998.

\bibitem{DBLP:conf/podc/Herlihy18}
Maurice Herlihy.
\newblock Atomic cross-chain swaps.
\newblock In Calvin Newport and Idit Keidar, editors, {\em Proceedings of the
  2018 {ACM} Symposium on Principles of Distributed Computing, {PODC} 2018,
  Egham, United Kingdom, July 23-27, 2018}, pages 245--254. {ACM}, 2018.

\bibitem{HW90}
Maurice~P. Herlihy and Jeannette~M. Wing.
\newblock Linearizability: a correctness condition for concurrent objects.
\newblock {\em ACM Transactions on Programming Languages and Systems (TOPLAS)},
  12(3):463--492, 1990.

\bibitem{Cosmos}
Tendermint Inc.
\newblock {Cosmos}.
\newblock \url{https://cosmos.network}.
\newblock [Online; accessed 22-November-2018].

\bibitem{DLT:health:2017}
Tsung-Ting Kuo, Hyeon-Eui Kim, and Lucila Ohno-Machado.
\newblock Blockchain distributed ledger technologies for biomedical and health
  care applications.
\newblock {\em Journal of the American Medical Informatics Association},
  24(6):1211--1220, 2017.

\bibitem{DBLP:conf/focs/MicaliRK03}
Silvio Micali, Michael~O. Rabin, and Joe Kilian.
\newblock Zero-knowledge sets.
\newblock In {\em 44th Symposium on Foundations of Computer Science {(FOCS}
  2003), 11-14 October 2003, Cambridge, MA, USA, Proceedings}, pages 80--91.
  {IEEE} Computer Society, 2003.
\newblock URL:
  \url{http://ieeexplore.ieee.org/xpl/mostRecentIssue.jsp?punumber=8767}.

\bibitem{DBLP:conf/srds/MilosevicHS11}
Zarko Milosevic, Martin Hutle, and Andr{\'{e}} Schiper.
\newblock On the reduction of atomic broadcast to consensus with byzantine
  faults.
\newblock In {\em SRDS 2011}, pages 235--244, 2011.

\bibitem{DBLP:journals/mst/MostefaouiPRJ17}
Achour Most{\'{e}}faoui, Matoula Petrolia, Michel Raynal, and Claude Jard.
\newblock Atomic read/write memory in signature-free byzantine asynchronous
  message-passing systems.
\newblock {\em Th. Comp. Syst.}, 60(4):677--694, 2017.

\bibitem{DBLP:conf/fc/MukhamedovKR05}
Aybek Mukhamedov, Steve Kremer, and Eike Ritter.
\newblock Analysis of a multi-party fair exchange protocol and formal proof of
  correctness in the strand space model.
\newblock In Andrew~S. Patrick and Moti Yung, editors, {\em Financial
  Cryptography and Data Security, 9th International Conference, {FC} 2005,
  Roseau, The Commonwealth of Dominica, February 28 - March 3, 2005, Revised
  Papers}, volume 3570 of {\em Lecture Notes in Computer Science}, pages
  255--269. Springer, 2005.

\bibitem{N08bitcoin}
Satoshi Nakamoto.
\newblock Bitcoin: A peer-to-peer electronic cash system.
\newblock \url{https://bitcoin.org/bitcoin.pdf}, 2008.
\newblock [Online; accessed 22-February-2020].

\bibitem{Namecoin}
Namecoin.
\newblock {Namecoin}.
\newblock \url{https://www.namecoin.org/}.
\newblock [Online; accessed 22-February-2020].

\bibitem{Oraclize}
Oraclize.
\newblock {Oraclize}.
\newblock \url{http://www.oraclize.it}.
\newblock [Online; accessed 22-November-2018].

\bibitem{osborne2004introduction}
Martin~J Osborne et~al.
\newblock {\em An introduction to game theory}, volume~3.
\newblock Oxford university press New York, 2004.

\bibitem{PolkaDot}
PolkaDot.
\newblock {PolkaDot}.
\newblock \url{https://polkadot.network}.
\newblock [Online; accessed 22-November-2018].

\bibitem{DLT:RealEstate:2017}
Avi Spielman.
\newblock {\em Blockchain: Digitally Rebuilding the Real Estate Industry}.
\newblock {MS} dissertation, Massachusetts Institute of Technology, 2016.

\bibitem{Zago2018}
Matteo~Gianpietro Zago.
\newblock {50+ Examples of How Blockchains are Taking Over the World}.
\newblock {\em Medium}, 2018.
\newblock URL:
  \url{https://medium.com/@matteozago/50-examples-of-how-blockchains-are-taking-over-the-world-4276bf488a4b}.

\end{thebibliography}

% \begin{thebibliography}{1}
% \bibitem{DLO_SIGACT18}
% A. Fern{\'{a}}ndez Anta, K.~M. Konwar, C. Georgiou, and
%   N.~C. Nicolaou.
% \newblock Formalizing and implementing distributed ledger objects.
% \newblock {\em {SIGACT} News}, 49(2):58--76, 2018.\vspace{-.5em}

% \bibitem{coelho2018byzantine}
% P. Coelho, T.~C. Junior, A. Bessani, F. Dotti, and
%   F. Pedone.
% \newblock Byzantine fault-tolerant atomic multicast.
% \newblock In {\em Proc. of IEEE DSN 2018}, pp. 39--50.\vspace{-.5em}

% \bibitem{CRISTIAN1995158}
% F.~Cristian, H.~Aghili, R.~Strong, and D.~Dolev.
% \newblock Atomic broadcast: From simple message diffusion to byzantine
%   agreement.
% \newblock {\em Information and Computation}, 118(1):158 -- 179, 1995.\vspace{-.5em}

% \bibitem{DBLP:conf/srds/MilosevicHS11}
% Z. Milosevic, M. Hutle, and A. Schiper.
% \newblock On the reduction of atomic broadcast to consensus with byzantine
%   faults.
% \newblock In {\em Proc. of IEEE SRDS 2011}, pp. 235--244.\vspace{-.5em}

% \bibitem{DBLP:journals/mst/MostefaouiPRJ17}
% A. Most{\'{e}}faoui, M. Petrolia, M. Raynal, and C. Jard.
% \newblock Atomic read/write memory in signature-free byzantine asynchronous
%   message-passing systems.
% \newblock {\em Theory Comput. Syst.}, 60(4):677--694, 2017.
% \end{thebibliography}

\end{document}